\documentclass[10pt]{my-elsarticle}

\usepackage{amssymb}
\usepackage{amsmath,amsthm}
\usepackage{a4wide}
\usepackage{hyperref}
\usepackage{caption}
\usepackage{mathrsfs}
\usepackage{siunitx}
\usepackage{url}

\newtheorem{theo}{Theorem}
\newtheorem{prop}{Proposition}

\newcommand{\eps}{\varepsilon}
\newcommand{\R}{\mathbb{R}}

\newcommand{\msE}{\mathscr{E}}
\newcommand{\msG}{\mathscr{G}}
\newcommand{\msN}{\mathscr{N}}
\newcommand{\msV}{\mathscr{V}}

\begin{document}

\begin{frontmatter}

\title{Complex network model for COVID-19: human behavior, pseudo-periodic solutions and multiple epidemic waves}

\author[1]{Cristiana J. Silva\corref{cor1}}
\ead{cjoaosilva@ua.pt}
\ead[url]{https://orcid.org/0000-0002-7238-546X}

\author[2]{Guillaume Cantin}
\ead{guillaumecantin@mail.com}
\ead[url]{https://orcid.org/0000-0001-5122-1194}

\author[1]{Carla Cruz}
\ead{carla.cruz@ua.pt}
\ead[url]{https://orcid.org/0000-0003-4082-7523}

\author[3]{Rui Fonseca-Pinto}
\ead{rui.pinto@ipleiria.pt}
\ead[url]{https://orcid.org/0000-0001-6774-5363}

\author[3,4]{Rui Passadouro da Fonseca}
\ead{rmfonseca@arscentro.min-saude.pt}
\ead[url]{https://orcid.org/0000-0002-7766-576X}

\author[4]{Estev\~{a}o Soares dos Santos}
\ead{essantos3@arscentro.min-saude.pt}
\ead[url]{https://orcid.org/0000-0001-6567-1487}

\author[1]{Delfim F. M. Torres}
\ead{delfim@ua.pt}
\ead[url]{https://orcid.org/0000-0001-8641-2505}


\cortext[cor1]{Corresponding author}

\address[1]{Center for Research and Development in Mathematics and Applications (CIDMA),\\
Department of Mathematics, University of Aveiro, 3810-193 Aveiro, Portugal}

\address[2]{Laboratoire de Math\'ematiques Appliqu\'ees, FR-CNRS-3335,\\ 
25, Rue Philippe Lebon, Le Havre Normandie 76063, France}

\address[3]{Center for Innovative Care and Health Technology (ciTechCare),
Polytechnic of Leiria, Portugal}

\address[4]{ACES Pinhal Litoral -- Central Regional Health Administration (ARS Centro), Portugal}


\begin{abstract}
We propose a mathematical model for the transmission dynamics of SARS-CoV-2 
in a homogeneously mixing non constant population, and generalize it to a model 
where the parameters are given by piecewise constant functions. This allows us 
to model the human behavior and the impact of public health policies on the dynamics 
of the curve of active infected individuals during a COVID-19 epidemic outbreak. 
After proving the existence and global asymptotic stability of the disease-free 
and endemic equilibrium points of the model with constant parameters, 
we consider a family of Cauchy problems, with piecewise constant parameters, 
and prove the existence of pseudo-oscillations between a neighborhood 
of the disease-free equilibrium and a neighborhood of the endemic equilibrium, 
in a biologically feasible region. In the context of the COVID-19 pandemic, 
this pseudo-periodic solutions are related to the emergence of epidemic waves. 
Then, to capture the impact of mobility in the dynamics of COVID-19 epidemics, 
we propose a complex network with six distinct regions based on COVID-19 real data 
from Portugal. We perform numerical simulations for the complex network model, 
where the objective is to determine a topology that minimizes the level of active 
infected individuals and the existence of topologies that are likely 
to worsen the level of infection. We claim that this methodology is a tool 
with enormous potential in the current pandemic context, and can be applied 
in the management of outbreaks (in regional terms) but also to manage 
the opening/closing of borders.
\end{abstract}




\end{frontmatter}


\section{Introduction}

The history of the pandemics that have already plagued mankind has revealed that there are consistently periods 
of marked increase in infected individuals, followed by phases in which the numbers are relatively lower. 
In these cases, there is also a repetition of these oscillations that are called pandemic waves.
A second wave of the pandemic poses an imminent threat to society, with an immense cost in terms of human lives 
and a devastating economic impact, it is therefore, crucial to try to avoid the emergence of pandemic waves. 

The coronavirus disease 2019 (COVID-19) is an infectious disease caused by the severe acute respiratory 
syndrome coronavirus-2 (SARS-CoV-2). The World Health Organization announced the COVID-19 outbreak 
as a pandemic on 11 March 2020 \cite{who:covid}. At the time of writing, several countries across 
Europe are seeing a resurgence in COVID-19 cases after successfully controlling the first outbreak 
and started to take action to face the so-called second wave \cite{second:wave:europe}. Mathematical 
tools have been important in the analysis and control of COVID-19 pandemic, see e.g. 
\cite{Steven:covid,Kankan,Ndairou:covid:wuhan,Moradian:JTMed} and references cited therein. 

In this paper, we propose a mathematical model given by a system of ordinary equations 
based on the model from \cite{Silvaetal:covidPT}, for the transmission dynamics of SARS-CoV-2 
in a homogeneously mixing population, and we generalize it to a non constant population model 
with piecewise constant parameters. Considering parameters determined by piecewise constant 
functions allows to model the governmental and public health decisions of political actors,
which have a large influence on the behaviors of individuals, which in turn
can change the dynamics of the epidemic, see e.g. 
\cite{ambrosio2020coupled,banos2015importance,epstein2008coupled}.
Moreover, the piecewise constant parameters also allows to mathematical 
model the human behavior in the application of non-pharmaceutical interventions (NPI), 
such as, physical distancing, limited size of indoor and outdoor gatherings, teleworking, 
regular cleaning of frequently-touched surfaces and appropriate ventilation of indoor spaces, 
mask use, avoiding close contact and hand washing, see e.g. 
\cite{CDC:preventive:measures,ECDC:guildelines}.  

Viruses need cells to replicate themselves, as they are not able to do so on their own. 
Depending on the type of virus, contagion is more or less facilitated depending on the means 
of transmission. In the case of SARS-COV-2, the transmission is carried out through droplets 
containing viruses that are sent from the respiratory tract. Accordingly, as the airways 
are the gateway to the virus, the proximity of people and the dynamics imposed by the movement 
are a determining factor in the spread of the COVID-19 disease. The proof that the movement 
of people is a determining factor for contagion is that during the recent periods of confinement 
the number of infections has decreased very markedly \cite{Vinceti}. 

In this paper, the COVID-19 mathematical model is improved by constructing a complex 
network of dynamical systems, following the framework presented in 
\cite{cantin2017nonidentical,cantin2019influence},
in order to take into account the mobilities of individuals,
which are also known to play a decisive role in the dynamics of the epidemic.
We investigate in which cases such a COVID-19 second wave can occur,
by establishing sufficient conditions for pseudo-periodic solutions: see, e.g.,
\cite{denu2020existence,wang2010travelling,hsu2019stability} for traveling waves 
results in epidemiological models and 
\cite{aziz2006synchronization,belykh2004connection,golubitsky2006nonlinear} 
for complex network models.  

Our model is calibrated in order to fit with the real data of the COVID-19 dynamics 
in six regions of Portugal mainland, namely \emph{Norte}, \emph{Centro}, 
\emph{Lisboa e Vale do Tejo}, \emph{Alentejo}, \emph{Algarve} 
and \emph{Pinhal Litoral}. 

Numerical simulations are provided where we explore the effect of the topology 
of the network on the dynamics of the epidemics (disposal of connections 
and coupling strength). We identify which type of topology minimizes the level 
of infection of the epidemic, and which type of topology worsens the number 
of infected individuals.

This paper is organized as follows. In Section~\ref{sec:model:constparam}, 
we propose a mathematical model for the transmission dynamics of SARS-COV-2, 
with constant parameters and variable population size. We show that the model 
admits two equilibrium points and analyze their local and global stability. 
In Section~\ref{sec:model:piecewise}, we consider piecewise constant parameters, 
which allows to model the impact of public health policies and the human behavior 
in the dynamics of the COVID-19 epidemic. The existence and uniqueness of global 
solutions of the model with piecewise constant parameters is proved. An important 
result in the context of the COVID-19 pandemic and the resurgence of epidemic waves, 
is proved in Section~\ref{sec:exist:pseudoperiodic}, where a sufficient condition 
is proved to the existence of pseudo-periodic solutions. 
In Section~\ref{sec:complexnet:model}, we construct a complex network of $SAIRP$ 
models with piecewise constant parameters, where each node represent one of six 
regions in Portugal, and where the values of the parametes differ from one region 
to another.  The case study, with COVID-19 real data from Portugal, is analyzed 
in Section~\ref{sec:casestudyPT}, where we calibrate the model to each of the six 
regions and after, numerical simulations are performed in order to determine 
the topologies that minimizes the average number of active infected individuals 
in Portugal. The results obtained in this paper are discussed, 
in Section~\ref{sec:discussion}, from a practical point of view and its implications 
in the management of the COVID-19 pandemic waves. We end the paper 
in Section~\ref{sec:conclusion}, with some conclusions and future work.   


\section{Model with constant parameters}
\label{sec:model:constparam}

We propose a compartmental $SAIRP$ mathematical model, based on \cite{Silvaetal:covidPT}, 
where the population is subdivided into five classes: susceptible individuals ($S$); 
asymptomatic infected individuals ($A$); active infected individuals ($I$); removed 
(including recovered and COVID-19 induced deaths) ($R$); protected individuals ($P$). 
The total population, $N(t) = S(t) + A(t) + I(t) + R(t) + P(t)$, with $t \in [0, T]$ 
representing the time (in days) and $T > 0$, has a variable size where the recruitment 
rate, $\Lambda$, and the natural death rate, $\mu > 0$, are assumed to be constant. 
The susceptible individuals $S$ become infected by contact with active infected $I$ 
and asymptomatic infected $A$ individuals, at a rate of infection 
$\beta \frac{\left( \theta A  + I \right)}{N}$, where $\theta$ represents a modification 
parameter for the infectiousness of the asymptomatic infected individuals $A$.
The remaining assumptions follows the one from \cite{Silvaetal:covidPT}: 
only a fraction $q$ of asymptomatic infected individuals $A$ develop symptoms 
and are detected, at a rate $v$. Active infected individuals $I$ are transferred 
to the recovered/removed individuals $R$, at a rate $\delta$, by recovery from the 
disease or by COVID-19 induced death. A fraction $p$, with $0 < p < 1$, is protected 
(not immune) from infection, by the application of non-pharmaceutical interventions (NPI), 
such as, physical distancing, limited size of indoor and outdoor gatherings, teleworking, 
regular cleaning of frequently-touched surfaces and appropriate ventilation of indoor 
spaces, mask use and hand washing, see e.g. \cite{CDC:preventive:measures,ECDC:guildelines}, 
that prevent from being exposed to the infection, and is transferred to the class of protected 
individuals $P$, at a rate $\phi$. A fraction $m$ of protected individuals $P$ returns 
to the susceptible class $S$, at a rate $w$. In what follows, for the sake of simplification, 
we use the notation $\nu = v q$ and $\omega = w m$. 
The previous assumptions are described by the following system:
\begin{equation}
\label{eq:model}
\begin{cases}
\dot{S}(t) = \Lambda - \beta (1-p)\frac{ \left( \theta A(t) 
+ I(t) \right)}{N(t)} S(t)  - \phi p S(t) + \omega P(t) -\mu S(t) ,\\[0.2 cm]
\dot{A}(t) = \beta (1-p) \frac{\left( \theta A(t) + I(t) \right)}{N(t)} S(t) 
- \nu  A(t) - \mu A(t), \\[0.2 cm]
\dot{I}(t) = \nu A(t) - \delta  I(t) - \mu I(t) ,\\[0.2 cm]
\dot{R}(t) = \delta I(t) - \mu R(t), \\[0.2 cm]
\dot{P}(t) = \phi p S(t) - \omega P(t) - \mu P(t) . 
\end{cases}
\end{equation}


\subsection{Existence, positivity and boundedness of solutions}

The equations of the SAIRP model \eqref{eq:model} 
can be rewritten as 
\begin{equation}
\label{eq:model-short-form}
\dot{x}(t) = f\left(x(t),\,\alpha \right),
\quad t > 0,
\end{equation}
with $x = (S,\,A,\,I,\,R,\,P)^T \in \R^5$ and 
$\alpha = (\Lambda,\,\mu,\,\beta,\,p,\,\theta,\,\phi,\,\omega,\,\nu,\,\delta)^T \in \R^9$,
where the non-linear operator $f$ is defined in $\R^5 \times \R^9$ by
\begin{equation}
f(x,\,\alpha) =
\begin{pmatrix}
\Lambda - \beta (1-p)\frac{ \left( \theta A + I\right)}{N} - \phi p S + \omega P -\mu S\\[0.2 cm]
\beta (1-p) \frac{\left( \theta A + I \right)}{N} S - \nu  A - \mu A\\[0.2 cm]
\nu A - \delta  I - \mu I\\[0.2 cm]
\delta I - \mu R \\[0.2 cm]
\phi p S - \omega P - \mu P
\end{pmatrix}.
\end{equation}

In order to prove that the problem determined by \eqref{eq:model-short-form} 
is well-posed, we introduce the compact region $\Omega \subset \R^5$ defined by
\begin{equation}
\label{eq:Omega}
\Omega = \left\lbrace x = (S,\,A,\,I,\,R,\,P)^T \in \big(\R^+\big)^5~;~
0 < S+A+I+R+P \leq \frac{\Lambda}{\mu} \right\rbrace.
\end{equation}

The following theorem establishes the existence of global 
solutions to \eqref{eq:model-short-form}.

\begin{theo}
\label{theo:well-posedness-model}
For any $x_0 = (S_0,\,A_0,\,I_0,\,R_0,\,P_0)^T \in \Omega$,
the Cauchy problem given by \eqref{eq:model-short-form} and $x(0) = x_0$ 
admits a unique solution, denoted by $x(t,\,x_0)$, defined on $[0,\,\infty)$,
whose components are non-negative. Furthermore, the region $\Omega$ defined 
by \eqref{eq:Omega} is positively invariant.
\end{theo}

\begin{proof}
The existence of a local in time solution $x(t,\,x_0)$ to problem \eqref{eq:model-short-form} 
starting from $x_0 \in \Omega$ follows from the theory of ordinary equations 
(see, for instance, \cite{perko2013differential}).
The non-negativity of the components is guaranteed by the quasi-positivity 
of the non-linear operator $f = (f_j)_{1 \leq j \leq 5}$,
which means that it satisfies the property
\begin{equation*}
f_i(x_1,\,\dots,\,x_{i-1},\,0,\,x_{i+1},\,\dots,\,x_{5},\,\alpha) \geq 0,
\end{equation*}
for all $x = (x_1,\,\dots,\,x_{5}) \in (\R^+)^{5}$, 
$i \in \lbrace 1,\,\dots,\,5 \rbrace$ and $\alpha\in \R^9$.
By virtue of Proposition A.17 in \cite{smith2011dynamical}, it follows 
that the components of any solution $x(t,\,x_0)$
stemming from $x_0$ in $\Omega$
remain non-negative in future time.
Finally, summing the five equations of system \eqref{eq:model-short-form} 
leads to
\begin{equation*}
\dot{N}(t) + \mu N(t) \leq \Lambda,
\quad
t>0,	
\end{equation*}
from which it is deduced, using Gronwall's lemma, that
\begin{equation*}
N(t) \leq \left[N(0) - \frac{\Lambda}{\mu} \right]e^{-\mu t} + \frac{\Lambda}{\mu},
\quad
t>0.
\end{equation*}
The latter inequality proves that the region $\Omega$ defined 
by \eqref{eq:Omega} is positively invariant.
\end{proof}


\subsection{Equilibrium points and basic reproduction number}

In this section, we study the equilibrium points of system \eqref{eq:model}
and derive their expressions with respect to the parameters of the model. 
We compute the basic reproduction number, denoted by $R_0$, 
following the approach in \cite{van2002reproduction}.  

\begin{prop}
The model \eqref{eq:model} has two equilibrium points: 
\begin{itemize}
\item disease-free equilibrium, denoted by $\Sigma_0$, given by
\begin{equation}
\label{eq:DFE}
\Sigma_0 = \left( S_0, A_0, I_0, R_0, P_0 \right) = 
\left( {\frac {\Lambda\, \left( \omega+\mu \right) }{\mu\, 
\left( p\phi+\mu+\omega \right) }}, 0, 0, 0, 
{\frac {\phi\,p\Lambda}{\mu\, \left( p\phi+
\mu+\omega \right) }} \right) \, ;
\end{equation}
\item endemic equilibrium, $\Sigma_+$, whenever $R_0 > 1$, given by
\begin{equation}
\Sigma_+ = \left( S_+, A_+, I_+, R_+, P_+ \right) 
\end{equation}
with
\begin{equation}
\label{eq:EE}
\begin{split}
S_+ &= \frac{ \Lambda (\omega+\mu)}{(p \phi+\mu+\omega) \mu} R_0^{-1} \, , \\
A_+ &= \frac{\Lambda }{ \nu+\mu}R_0^{-1} (R_0-1) \, ,\\
I_+ &=\frac{\Lambda \nu }{(\nu+\mu) (\delta+\mu)}R_0^{-1} (R_0-1) \, ,\\
R_+ &=\frac{\delta \Lambda \nu  }{ (\nu+\mu) (\delta+\mu) \mu}R_0^{-1} (R_0-1) \, ,\\
P_+ &= \frac{ \Lambda \phi p}{(p \phi+\mu+\omega)  \mu}R_0^{-1} \, ,
\end{split}
\end{equation}
\end{itemize}
where the basic reproduction number, $R_0$, is given by 
\begin{equation}
R_0 = {\frac { \beta\, \left( 1-p \right)
\left( \delta\,\theta+\mu\,\theta+\nu \right) 
\left( \omega +\mu \right)  }{ \left( \delta+\mu \right) 
\left( \nu+\mu \right)  \left( p\phi+\mu+\omega \right) }}
= \frac{\mathcal{N}}{\mathcal{D}} \, .
\end{equation}
\end{prop}

\begin{proof}
The computation of the equilibrium points and the basic reproduction 
number follows standard arguments, see \cite{van2002reproduction}. 
\end{proof}


\subsection{Local and global stability analysis}

In what follows we prove the local and global asymptotic stability of the 
disease-free equilibrium (DFE), when the basic reproduction number satisfies 
the inequality $R_0 < 1$. For $R_0 > 1$, we prove that the endemic equilibrium 
is globally asymptotically stable in a biologically meaningful 
compact positive invariant region. 

\begin{theo}[local stability of the DFE]
The disease-free equilibrium, $\Sigma_0$, is locally asymptotically 
stable whenever $R_0 < 1$. 
\end{theo}

\begin{proof}
The Jacobian matrix of system \eqref{eq:model}, evaluated at the disease-free equilibrium 
\eqref{eq:DFE}, is given by
\begin{equation*}
M\left(\Sigma_0\right) = 
\left[ 
\begin {array}{ccccc} 
-(\phi\,p+\mu) & -\frac{\theta\,\beta\,\left( \mu+\omega \right) \left(1-p \right) }{\phi\,p+\mu+\omega} 
& -\frac{\beta\, \left( \mu+\omega \right)  \left(1-p \right) }{\phi	\,p+\mu+\omega} & 0 & \omega\\ 
\noalign{\medskip} 0 & -\frac{\beta\theta(1-p)(\mu+\omega)+(\mu+\nu)(p\phi+\mu+\omega)}{\phi\,p+\mu+\omega} 
& \frac {\beta\, \left( \mu+\omega \right)  \left(1-p \right) }{\phi\,p+\mu+\omega}&0&0\\
\noalign{\medskip}0&\nu&-\delta-\mu&0&0\\ 
\noalign{\medskip}0&0&\delta&-\mu&0\\ 
\noalign{\medskip}\phi\,p&0&0&0&-(\mu+\omega)
\end{array} \right]  \, .
\end{equation*}
The eigenvalues of the matrix $M\left(\Sigma_0\right)$ are given by 
$\lambda_1 = \lambda_2 = -\mu$, $\lambda_3 = -(\phi p + \mu +\omega)$ 
and the remaining two, $\lambda_4$ and $\lambda_5$, are the roots 
of the polynomial $p(\lambda)$ given by
\begin{equation*}
p(\lambda) = \lambda^2 + B \lambda + C ,
\end{equation*}
where
$B = \frac{-\beta \theta (1-p) (\omega+\mu)}{(p \phi+\mu+\omega)} +\delta+2 \mu+\nu$ and 
$C = \frac{\mathcal{D-\mathcal{N}}}{p \phi+\mu+\omega}$. Applying the Routh--Hurwitz criterion, 
we conclude that model \eqref{eq:model} is locally stable if, and only if, $A > 0$ and $B > 0$.
It is easy to show that $C > 0$ whenever $R_0 < 1$. The coefficient $B$ is positive when
$$
\beta \theta (1-p) (\omega+\mu)   < (\delta+2 \mu+\nu)(p \phi+\mu+\omega) \, .
$$ 
Since all parameters take positive values, we have
\begin{equation*}
\beta \theta (1-p) (\omega+\mu) < \underbrace{\beta\, \left( 1-p \right)
\left( \delta\,\theta+\mu\,\theta+\nu \right) 
\left( \omega +\mu \right)}_{\mathcal{N}} \, .
\end{equation*}
From $R_0 < 1$, we have ${\mathcal{N}} < {\mathcal{D}}$ and, therefore,
\begin{equation*}
\underbrace{\beta\, \left( 1-p \right)
\left( \delta\,\theta+\mu\,\theta+\nu \right) 
\left( \omega +\mu \right)}_{\mathcal{N}} < \underbrace{\left( \delta+\mu \right) 
\left( \nu+\mu \right)  \left( p\phi+\mu+\omega \right)}_\mathcal{D} \, .
\end{equation*}
From the biological meaning of $\delta$, $\mu$ and $\nu$, we have that $0 < \delta, \mu, \nu < 1$ 
(observe that their unit is $day \, ^{-1}$). Therefore, 
\begin{equation*}
0 < \delta + \mu < 2 \quad \text{and} \quad 0 < \mu + \nu < 2
\end{equation*}
and
\begin{equation*}
\frac{1}{\delta + \mu} > \frac{1}{2} \quad \text{and} 
\quad \frac{1}{\mu + \nu}  > \frac{1}{2},
\end{equation*}
that is, 
\begin{equation*}
\frac{1}{\delta + \mu} + \frac{1}{\mu + \nu}  > 1 \, .
\end{equation*}
Multiplying the previous inequality by $(\delta + \mu) (\mu + \nu)$ we have 
\begin{equation*}
\left( \frac{1}{\delta + \mu} + \frac{1}{\mu + \nu} \right) (\delta + \mu) (\mu + \nu) 
> (\delta + \mu) (\mu + \nu) \, .
\end{equation*}
In other words, 
\begin{equation*}
(\delta + \mu) (\mu + \nu) < \delta + 2 \mu + \nu  \, .
\end{equation*}
Therefore, 
\begin{equation*} 
\left( \delta+\mu \right) \left( \nu+\mu \right)  \left( p\phi+\mu+\omega \right) 
< \left( \delta + 2 \mu + \nu \right)  \left( p\phi+\mu+\omega \right),
\end{equation*}
which ensures that $B > 0$. It follows the local asymptotic stability 
of the disease-free equilibrium $\Sigma_0$.  
\end{proof}

\begin{theo}[global stability  of the DFE]
\label{theo:glob:DFE}
If $R_0 < 1$, then the disease-free equilibrium, $\Sigma_0$, 
is globally asymptotically stable in $\Omega$. 
\end{theo}

\begin{proof}
Since $R_0 < 1$, we can write $R_0 = 1-\eta$ with $0 < \eta < 1$. 
We obtain that
\[
\frac{\beta\, \left( 1-p \right)\left( \delta\,\theta
+\mu\,\theta+\nu \right)\left( \omega +\mu \right)}
{\left( \delta+\mu \right) \left( \nu+\mu \right)  \left( p\phi+\mu+\omega \right)}
= 1-\eta,
\]
which leads to
\begin{equation}
\label{eq:k(1-eta)}
\frac{\beta\, \left( 1-p \right)\left( \delta\,\theta
+\mu\,\theta+\nu \right)}{\left( \delta+\mu \right) \left( \nu+\mu \right)}
= k(1-\eta),
\end{equation}
where $k$ is defined by
\[
k = \dfrac{p\phi+\mu+\omega}{\mu+\omega}.
\]
\emph{A fortiori}, we have:
\begin{equation}
\label{eq:zeta>0}
\dfrac{\beta(1-p)\theta}{\nu+\mu} < k(1-\eta).
\end{equation}
Now we consider the following functional given by
\[
L = S-S_0 - S_0 \ln\frac{S}{S_0}
+ A
+ \zeta I
+ \xi \left(P-P_0 - P_0 \ln\frac{P}{P_0}\right)
+ \chi \left(N-N_0 - N_0 \ln\frac{N}{N_0}\right),
\]
where $\zeta$ and $\xi$ are defined by
\begin{equation}
\label{eq:zeta-xi}
\zeta = \dfrac{k(1-\eta)(\nu+\mu)-\beta(1-p)\theta}{k\nu(1-\eta)},
\quad \xi = \dfrac{\omega P_0}{\phi pS_0},
\end{equation}
and $\chi$ is a positive constant which will be determined below.
Recall that $N_0 = \frac{\Lambda}{\mu}$,
and note that $\zeta > 0$ by virtue of \eqref{eq:zeta>0}.
As constructed, $L$ is a non-negative functional and we have
\[
L = 0 \iff (S, A, I, R, P) = \Sigma_0.
\]
Now we compute the derivative of $L$ along the solutions of the SAIRP 
model \eqref{eq:model} starting in $\Omega$. We have
\[
\begin{split}
\dot{L}
&= \left(1-\dfrac{S_0}{S}\right)\dot{S}
+ \dot{A}
+ \zeta \dot{I}
+ \xi \left(1-\dfrac{P_0}{P}\right)\dot{P}
+ \chi \left(1-\dfrac{N_0}{N}\right)\dot{N}\\
&= \left(1-\dfrac{S_0}{S}\right)\left[\Lambda 
- \beta(1-p)\dfrac{\theta A + I}{N}S - (p\phi+\mu)S + \omega P \right]\\
&+ \left[\beta(1-p)\dfrac{\theta A + I}{N}S 
- (\nu+\mu)A\right]	+ \zeta\left[\nu A - (\delta+\mu)I\right] \\
&+ \xi \left(1-\dfrac{P_0}{P}\right)\left[\phi p S - (\omega+\mu)P\right] 
+ \chi \left(1-\dfrac{N_0}{N}\right)(\Lambda - \mu N).
\end{split}
\]
Now we use the relations
\[
\Lambda = (p\phi+\mu)S_0 - \omega P_0,
\quad p\phi S_0 = (\omega+\mu)P_0
\]
to obtain
\[
\begin{split}
\dot{L} &= \left(1-\dfrac{S_0}{S}\right)\left[- \beta(1-p)
\dfrac{\theta A + I}{N}S - (p\phi+\mu)(S-S_0) + \omega (P-P_0) \right] \\
&+ \left[\beta(1-p)\dfrac{\theta A + I}{N}S - (\nu+\mu)A\right]	
+ \zeta\left[\nu A - (\delta+\mu)I\right] \\
&+ \xi \left(1-\dfrac{P_0}{P}\right)\left[\phi p (S-S_0) 
- (\omega+\mu)(P-P_0)\right] + \chi \left(1-\dfrac{N_0}{N}\right)(\Lambda - \mu N).
\end{split}
\]
\emph{First step.} Let us examine the terms involving $(S-S_0)$ and $(P-P_0)$. 
We have
\[
\begin{split}
\Theta_1 =
&-(p\phi+\mu)\left(1-\dfrac{S_0}{S}\right)(S-S_0) + \omega \left(1-\dfrac{S_0}{S}\right)(P-P_0)	\\
&+\xi p\phi\left(1-\dfrac{P_0}{P}\right)(S-S_0) - \xi(\omega+\mu)\left(1-\dfrac{P_0}{P}\right)(P-P_0)\\
=	&-(p\phi+\mu)S_0\left(1-\dfrac{S_0}{S}\right)\left(\dfrac{S}{S_0}-1\right) 
+ \omega P_0\left(1-\dfrac{S_0}{S}\right)\left(\dfrac{P}{P_0}-1\right)\\
&+\xi p\phi S_0\left(1-\dfrac{P_0}{P}\right)\left(\dfrac{S}{S_0}-1\right) 
- \xi(\omega+\mu)P_0\left(1-\dfrac{P_0}{P}\right)\left(\dfrac{P}{P_0}-1\right)\\
=	&-(p\phi+\mu)S_0\left(\dfrac{S_0}{S}+\dfrac{S}{S_0}-2\right)
+ \omega P_0\left[\left(\dfrac{P}{P_0}+\dfrac{S_0}{S}-2\right)
+\left(1-\dfrac{P}{P_0}\times\dfrac{S_0}{S}\right)\right] \\
&+\xi p\phi S_0\left[\left(\dfrac{S}{S_0}+\dfrac{P_0}{P}-2\right)
+\left(1-\dfrac{S}{S_0}\times\dfrac{P_0}{P}\right)\right]	
- \xi(\omega+\mu)P_0\left(\dfrac{P_0}{P}+\dfrac{P}{P_0}-2\right).
\end{split}
\]
Now we have chosen $\xi$ such that $\omega P_0 = \xi p\phi S_0$.
The latter equality becomes
\[
\begin{split}
\Theta_1 =
&-(p\phi + \mu)S_0\left(\dfrac{S_0}{S}+\dfrac{S}{S_0}-2\right)
+\omega P_0 \left(\dfrac{P}{P_0}+\dfrac{S_0}{S}-2\right)
+\omega P_0 \left(1-\dfrac{P}{P_0}\times\dfrac{S_0}{S}\right) \\
&+\omega P_0 \left(1-\dfrac{S}{S_0}\times\dfrac{P_0}{P}\right)
+\omega P_0 \left(\dfrac{S}{S_0}+\dfrac{P_0}{P}-2\right)
-\xi(\omega+\mu)P_0\left(\dfrac{P_0}{P}+\dfrac{P}{P_0}-2\right)	\\
=&\left[-(p\phi + \mu)S_0 + \omega P_0\right] \left(\dfrac{S_0}{S}+\dfrac{S}{S_0}-2\right)	
+\left[\omega P_0 - \xi(\omega+\mu)P_0\right] \left(\dfrac{P_0}{P}+\dfrac{P}{P_0}-2\right)	\\
&+ \omega P_0 \left(2 - \dfrac{S_0 S^{-1}}{P_0 P^{-1}} - \dfrac{P_0 P^{-1}}{S_0 S^{-1}} \right).
\end{split}
\]
Elementary computations show that
\[
-(p\phi + \mu)S_0 + \omega P_0 < 0,
\quad \omega P_0 - \xi(\omega+\mu)P_0 = 0.
\]
Now we use the standard inequality
\[
2 - x - \dfrac{1}{x} \leq 0,
\quad \forall x > 0,
\]
to conclude that $\Theta_1 \leq 0$.
Thus, we obtain
\[
\begin{split}
\dot{L}
&\leq \left(1-\dfrac{S_0}{S}\right)\left[- \beta(1-p)\dfrac{\theta A + I}{N}S \right]
+ \left[\beta(1-p)\dfrac{\theta A + I}{N}S - (\nu+\mu)A\right] \\
&+ \zeta\left[\nu A - (\delta+\mu)I\right] + \chi \left(1-\dfrac{N_0}{N}\right)(\Lambda - \mu N).
\end{split}
\]
\emph{Second step.}
We examine the terms involving $A$ and $I$. After simplifications, we have
\[
\begin{split}
\dot{L}
&\leq
\beta (1-p) \theta A \left[\dfrac{S_0}{N} - \dfrac{(\nu+\mu)-\zeta \nu}{\beta (1-p) \theta} \right]
+\beta (1-p) I \left[\dfrac{S_0}{N} - \dfrac{\zeta (\delta+\mu)}{\beta (1-p)} \right]
- \chi \mu \dfrac{(N-N_0)^2}{N}	\\
&\leq \beta (1-p) \theta A \left[\dfrac{S_0}{N_0} 
- \dfrac{(\nu+\mu)-\zeta \nu}{\beta (1-p) \theta} \right]
+\beta (1-p) I \left[\dfrac{S_0}{N_0} - \dfrac{\zeta (\delta+\mu)}{\beta (1-p)} \right]	\\
&+\beta (1-p)\theta A \left(\dfrac{S_0}{N} - \dfrac{S_0}{N_0}\right) 
- \dfrac{\chi}{2} \mu \dfrac{(N-N_0)^2}{N} \\
&+\beta (1-p) I \left(\dfrac{S_0}{N} - \dfrac{S_0}{N_0}\right) 
- \dfrac{\chi}{2} \mu \dfrac{(N-N_0)^2}{N}.
\end{split}
\]
Next, we write
\[
\Theta_2 = \beta (1-p)\theta A \left(\dfrac{S_0}{N} - \dfrac{S_0}{N_0}\right) 
- \dfrac{\chi}{2} \mu \dfrac{(N-N_0)^2}{N} 
= \dfrac{\beta (1-p)\theta S_0}{N N_0} A \left(N_0-N\right) 
- \dfrac{\chi}{2} \mu \dfrac{(N-N_0)^2}{N}.
\]
By virtue of Young's inequality, we have
\[
A ( N_0 - N ) \leq \dfrac{\eps A^2}{2} + \dfrac{(N_0-N)^2}{2\eps},
\]
where $\eps$ is a positive coefficient which can be chosen arbitrarily small.
It follows that
\[
\begin{split}
\Theta_2 &\leq \dfrac{\beta (1-p)\theta S_0}{N N_0} \left[
\dfrac{\eps A^2}{2} + \dfrac{(N_0-N)^2}{2\eps}\right]
- \dfrac{\chi}{2} \mu \dfrac{(N-N_0)^2}{N}\\
&\leq \beta (1-p)\theta \dfrac{A}{N}\times A \times \dfrac{\eps S_0}{2 N_0}
+ \left[\dfrac{\beta (1-p)\theta S_0}{2 N_0 \eps} 
- \dfrac{\chi}{2}\mu \right] \dfrac{(N-N_0)^2}{N}.
\end{split}
\]
Since $A \leq N$, we obtain
\[
\Theta_2 \leq \beta (1-p)\theta \times A \times \dfrac{\eps S_0}{2 N_0}
+ \left[\dfrac{\beta (1-p)\theta S_0}{2 N_0 \eps} - \dfrac{\chi}{2}\mu \right] 
\dfrac{(N-N_0)^2}{N}.
\]
Similar computations show that
\[
\beta(1-p) I \left(\dfrac{S_0}{N} - \dfrac{S_0}{N_0}\right) 
- \dfrac{\chi}{2} \mu \dfrac{(N-N_0)^2}{N}
\leq \beta(1-p) \times I \times \dfrac{\eps S_0}{2 N_0}
+ \left[\dfrac{\beta(1-p) S_0}{2 N_0 \eps} 
- \dfrac{\chi}{2}\mu \right] \dfrac{(N-N_0)^2}{N}.
\]
We obtain that
\[
\begin{split}
\dot{L} &\leq \beta (1-p) \theta A \left[\dfrac{S_0}{N_0}\left(1+\dfrac{\eps}{2}\right) 
- \dfrac{(\nu+\mu)-\zeta \nu}{\beta (1-p) \theta} \right]
+ \beta (1-p) I \left[\dfrac{S_0}{N_0}\left(1+\dfrac{\eps}{2}\right) 
- \dfrac{\zeta (\delta+\mu)}{\beta (1-p)} \right] \\
&+ \left[\dfrac{\beta (1-p)\theta S_0}{2 N_0 \eps} 
- \dfrac{\chi}{2}\mu \right] \dfrac{(N-N_0)^2}{N}
+ \left[\dfrac{\beta (1-p) S_0}{2 N_0 \eps} 
- \dfrac{\chi}{2}\mu \right] \dfrac{(N-N_0)^2}{N}.
\end{split}
\]
Now we choose
\[
\eps = \dfrac{2\eta}{1-\eta} = \dfrac{2(1-R_0)}{R_0} > 0,
\]
which yields $1+\frac{\eps}{2} = \frac{1}{1-\eta}$.
Since $\frac{S_0}{N_0} = \frac{1}{k}$, we obtain, 
by virtue of \eqref{eq:zeta-xi},
\[
\dfrac{S_0}{N_0}\left(1+\dfrac{\eps}{2}\right) 
- \dfrac{(\nu+\mu)-\zeta \nu}{\beta (1-p) \theta}
= \dfrac{1}{k(1-\eta)} - \dfrac{(\nu+\mu)-\zeta \nu}{\beta (1-p) \theta} = 0.
\]
Analogously, it is seen that
\[
\begin{split}
\dfrac{S_0}{N_0}\left(1+\dfrac{\eps}{2}\right) - \dfrac{\zeta (\delta+\mu)}{\beta (1-p)}
&= \dfrac{1}{k(1-\eta)} - \dfrac{\zeta (\delta+\mu)}{\beta (1-p)}\\
&= \dfrac{\beta (1-p) (\theta \delta + \theta \mu + \nu) 
- k(1-\eta)(\nu+\mu)(\delta+\mu)}{\beta (1-p) \nu k(1-\eta)} = 0,	
\end{split}												
\]
by virtue of \eqref{eq:k(1-eta)}. After those simplifications, we obtain
\[
\begin{split}
\dot{L} &\leq \left[\dfrac{\beta (1-p)\theta S_0}{2 N_0 \eps} 
- \dfrac{\chi}{2}\mu \right] \dfrac{(N-N_0)^2}{N}
+ \left[\dfrac{\beta(1-p) S_0}{2 N_0 \eps} 
- \dfrac{\chi}{2}\mu \right] \dfrac{(N-N_0)^2}{N} \\
&\leq \left[\dfrac{\beta(1-p) S_0}{N_0 \eps} 
- \chi\mu \right] \dfrac{(N-N_0)^2}{N},
\end{split}
\]
since $\theta \leq 1$. Finally, we choose $\chi > 0$ sufficiently large so that
\[
\dfrac{\beta (1-p) S_0}{N_0 \eps} - \chi\mu < 0,
\]
which guarantees that $\dot{L} \leq 0$. In other words, the functional $L$ 
is a Lyapunov function for the flow induced 
by the SAIRP model \eqref{eq:model}. The conclusion follows from LaSalle's 
invariance principle \cite{lasalle1960some}.
\end{proof}

The next theorem establishes the global stability of the endemic equilibrium (EE)
$\Sigma_+$ defined by \eqref{eq:EE}, in an invariant region $\Gamma \subset \Omega$.
It remains an open question to determine if the endemic equilibrium 
is globally asymptotically stable in the whole region $\Omega$.

\begin{theo}[global stability of the EE]
\label{theo:glob:EE}
The compact region $\Gamma$ defined by
\begin{equation}
\label{eq:Gamma-invariant-region}
\Gamma = \left\lbrace x = (S,\,A,\,I,\,R,\,P)^T \in \big(\R^+\big)^5~;  
S+A+I+R+P = \frac{\Lambda}{\mu} \right\rbrace
\end{equation}
is positively invariant under the flow induced by system \eqref{eq:model}.
It contains the disease-free equilibrium, 
$\Sigma_0$, and the endemic equilibrium, $\Sigma_+$, if $R_0 > 1$.
Furthermore, if $R_0 > 1$, then the endemic equilibrium $\Sigma_+$ 
is globally asymptotically stable in $\Gamma$.
\end{theo}

\begin{proof}
Let us first prove that $\Gamma$ is positively invariant under 
the flow induced by system \eqref{eq:model}. We denote by 
$x(t,\,x_0) = \big(S(t),\,A(t),\,I(t),\,R(t),\,P(t)\big)$ 
the solution of system \eqref{eq:model} starting from $x_0 \in \Gamma$.
Summing again the five equations of system \eqref{eq:model} leads to
\[
\dot{N}(t) + \mu N(t) = \Lambda,
\quad t>0,
\]
where $N(t) = S(t) + A(t) + I(t) + R(t) + P(t)$.
Consequently, we have
\[
N(t) = \left[N(0) - \frac{\Lambda}{\mu} \right]e^{-\mu t} 
+ \frac{\Lambda}{\mu}, \quad t>0.
\]
Now we have $N(0) = \frac{\Lambda}{\mu}$ since $x_0 \in \Gamma$.
It follows that $N(t) = \frac{\Lambda}{\mu}$ for all $t > 0$,
which proves that $\Gamma$ is positively invariant.
Next, we easily verify that $\Sigma_0$ and $\Sigma_+$ belong to $\Gamma$.
Now we turn to prove the global stability of $\Sigma_+$ in $\Gamma$. 
To that aim, we introduce the functional $V$ defined by
\begin{equation}
\label{eq:Lyapunov-functional-EE}
\begin{split}
V &= c_1 \left(S-S_+-S_+ \ln\frac{S}{S_+}\right)
+ c_2 \left(A-A_+-A_+ \ln\frac{A}{A_+}\right)\\
&+ c_3 \left(I-I_+-I_+ \ln\frac{I}{I_+}\right)
+ c_4 \left(P-P_+-P_+ \ln\frac{P}{P_+}\right),
\end{split}
\end{equation}
with positive coefficients $c_1$, $c_2$, $c_3$, and $c_4$ satisfying
\begin{equation}
\label{eq:coefficients-V}
c_1 = c_2,
\quad c_1 \omega P_+ = c_4 p\phi S_+,
\quad c_3 \nu A_+ = c_1 \beta (1-p)\dfrac{S_+ I_+}{N_+},
\end{equation}
where $N_+ = \frac{\Lambda}{\mu}$.
As constructed, $V$ is a non-negative functional and we have
\[
V = 0 \iff (S, A, I, R, P) = \Sigma_+.
\]
We compute the time derivative of $V$ along the solutions 
of system \eqref{eq:model} starting in $\Gamma$. Using the 
fact that $\Sigma_+$ is an equilibrium of system \eqref{eq:model}, 
we obtain, after simplifications, that
\[
\begin{split}
\dot{V} &= c_1 \left(1-\dfrac{S_+}{S}\right)\left[\beta (1-p) 
\dfrac{\theta A_+ + I_+}{N_+} S_+ - \beta (1-p) \dfrac{\theta A + I}{N_+} S
+ \omega (P-P+) - (p\phi+\mu)(S-S_+)\right] \\
&+ c_2 \left(1-\dfrac{A_+}{A}\right)\left[\beta (1-p) 
\dfrac{\theta A + I}{N_+} S - \beta (1-p) \dfrac{\theta A_+ + I_+}{N_+} S_+
- (\nu+\mu)(A-A_+)\right]\\
&+ c_3 \left(1-\dfrac{I_+}{I}\right)\left[\nu(A-A_+) - (\delta+\mu)(I-I_+)\right]\\
&+ c_4 \left(1-\dfrac{P_+}{P}\right)\left[p\phi(S-S_+) - (\omega+\mu)(P-P_+)\right].
\end{split}
\]
We perform elementary computations, so as to split $\dot{V}$ into four terms, as follows:
\[
\dot{V} = \Phi_1 + \Phi_2 + \Phi_3 + \Phi_4,
\]
with
\[
\begin{split}
\Phi_1 	&= c_1 \left(1-\dfrac{S_+}{S}\right)\left[\omega (P-P+)-(p\phi+\mu)(S-S_+)\right]\\
&+ c_4 \left(1-\dfrac{P_+}{P}\right)\left[p\phi(S-S_+)-(\omega+\mu)(P-P_+)\right],\\
\Phi_2 	&= -c_1\beta (1-p) \theta \dfrac{S_+ A_+}{N_+}\left(\dfrac{S}{S_+} + \dfrac{S_+}{S}-2\right),\\
\Phi_3	&= -c_3(\delta + \mu)I_+ \left(\dfrac{I}{I_+} + \dfrac{I_+}{I}-2\right)
+ c_3 \nu A_+ \left(1- \dfrac{I_+}{I}\right)\left(\dfrac{A}{A_+}-1\right)\\
&+ c_1 \dfrac{\beta (1-p)}{N_+}\left(1-\dfrac{S_+}{S}\right) (I_+ S_+ - IS)
+ c_2 \dfrac{\beta (1-p)}{N_+}\left(1-\dfrac{A_+}{A}\right) (I S - I_+ S_+),\\
\Phi_4 	&= \left[-c_2(\nu+\mu)A_+ + c_1 \beta (1-p)\theta 
\dfrac{S_+ A_+}{N_+}\right]\left(\dfrac{A}{A_+} + \dfrac{A_+}{A}-2\right).
\end{split}
\]
First, similar computations as in the proof of the global stability 
of the disease-free equilibrium $\Sigma_0$ show that $\Phi_1 \leq 0$, 
since $c_1 \omega P_+ = c_4 p\phi S_+$. Next, it is seen that $\Phi_2 \leq 0$, 
by virtue of the elementary inequality
\[
u + \frac{1}{u} - 2 \geq 0, \quad \forall u \in \R.
\]
Afterwards, remarking that $(\delta + \mu) I_+ = \nu A_+$, 
algebraic computations show that $\Phi_3 = \Psi_1 + \Psi_2$ with
\[
\begin{split}
&\Psi_1 = c_3 \nu A_+ \left(3-\dfrac{S_+}{S} - \dfrac{A I_+}{A_+ I} 
- \dfrac{A_+ I S}{A I_+ S_+}\right),\\
&\Psi_2 = c_3 \nu A_+ \left(\dfrac{A}{A_+} + \dfrac{A_+}{A} - 2\right).
\end{split}
\]
We introduce $u = \frac{S_+}{S}$, $v = \frac{A I_+}{A_+ I}$ 
and $w = \frac{A_+ I S}{A I_+ S_+}$.
It is observed that $u v w = 1$.
By virtue of the standard inequality
\[
\big(uvw\big)^{1/3} \leq \dfrac{u + v + w}{3},
\]
we deduce that $\Psi_1 \leq 0$.
Finally, using \eqref{eq:EE} and \eqref{eq:coefficients-V},
we show that $\Phi_4 + \Psi_2 = 0$.
Gathering the above results shows that $\dot{V} \leq 0$.
Therefore, we have proved that the functional $V$ is a Lyapunov 
function for the flow induced by the SAIRP model \eqref{eq:model}.
The conclusion follows once again from LaSalle's invariance principle \cite{lasalle1960some}.
\end{proof}

The theoretical results proved in Theorems~\ref{theo:glob:DFE} and \ref{theo:glob:EE} 
are illustrated in Figure~\ref{fig:DFE-EE-GAS}. 
\begin{figure}
\centering
\includegraphics[scale=0.75]{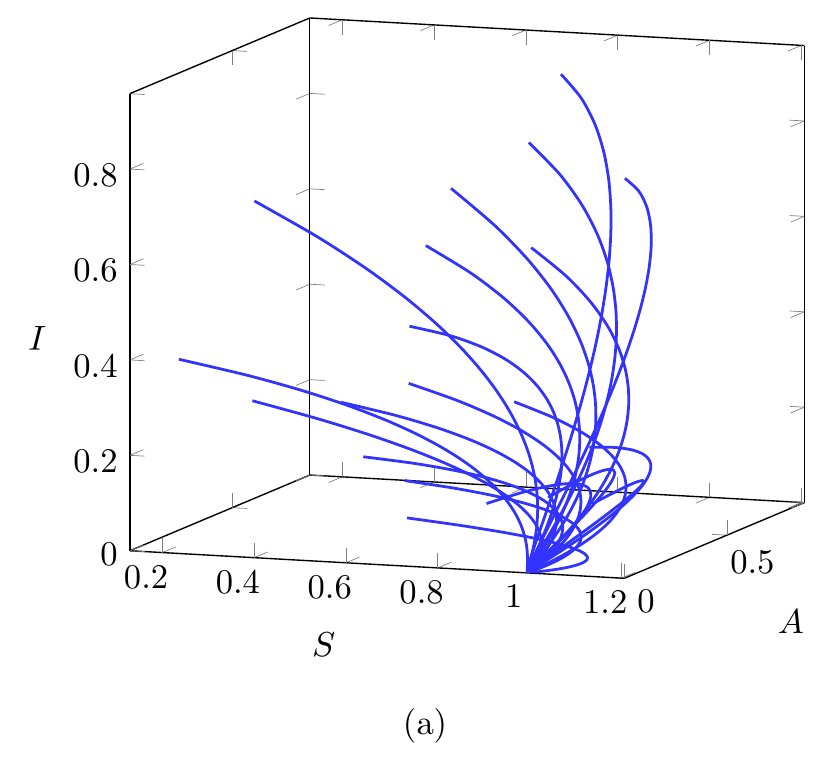}
\includegraphics[scale=0.75]{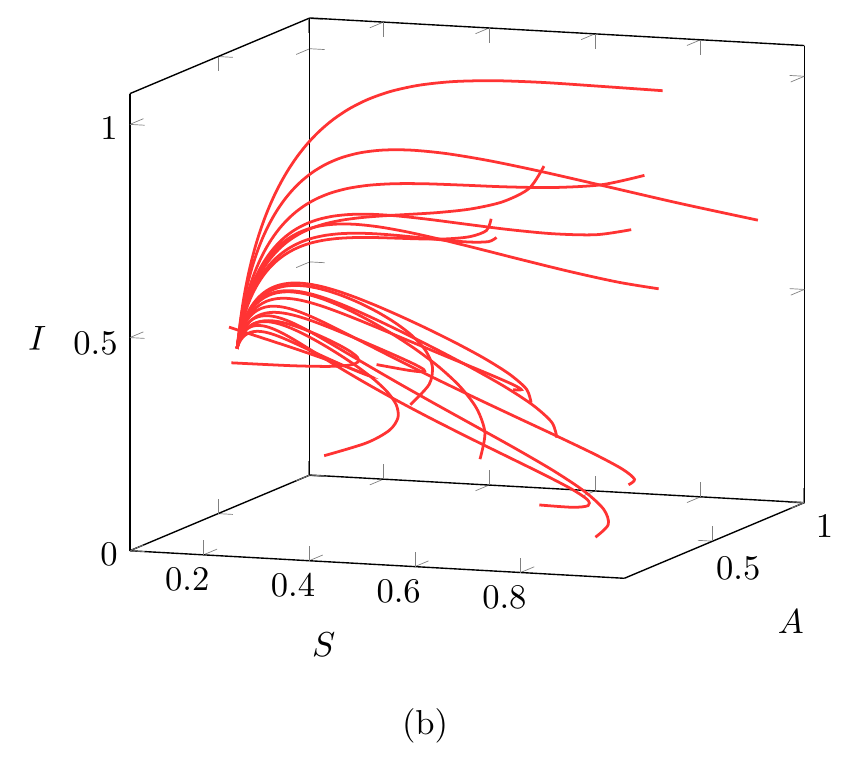}
\caption{Phase portraits in the $(S,\,A,\,I)$ space, illustrating 
the global stability of the equilibrium points.
(a) Global stability of the disease-free equilibrium ($R_0 < 1$).
(b) Global stability of the endemic equilibrium ($R_0 > 1$).}
\label{fig:DFE-EE-GAS}
\end{figure}


\section{Model with piecewise constant parameters}
\label{sec:model:piecewise}

The human behavior and the governmental public health decision makers 
can change the value of the basic reproduction number, $R_0$, and consequently 
the dynamics of model \eqref{eq:model}. In this section, we model the human 
behavior and the impact of the decision makers policies, by considering in model 
\eqref{eq:model} parameters determined by piecewise constant functions. 
We prove the existence and uniqueness of global solutions of the resulting model.
We start by subdividing the time line $[0,\,+\infty)$ into a finite number of $n$ intervals
\[
[T_0,\,T_1) \cup [T_1,\,T_2) \cup \dots \cup [T_n,\,+\infty),
\]
with disjoint unions, and introduce a piecewise constant 
function $\alpha$ defined on each time interval as
\[
\alpha(t) = \alpha_i,
\quad t \in [T_i,\,T_{i+1}),
\quad 0 \leq i \leq n,
\]
with $T_0=0$, $T_{n+1} = +\infty$ and $\alpha_i \in \R^9$.
Next, we consider the sequence of Cauchy problems defined 
for each initial condition $x_0 \in \Omega$ by
\begin{equation}
\label{eq:sequence-Cauchy-pb}
\begin{cases}
\begin{array}{lllll}
&x(0) = x_0, &\dot{x}(t) = f\big(x(t),\,\alpha_0\big),	
&T_0 < t < T_1,	& \\
&x(T_i) = \lim\limits_{\substack{t \to T_i\\
t \in (T_{i-1},\,T_i)}} x(t), &\dot{x}(t) = f\big(x(t),\,\alpha_i\big), 
&T_i < t < T_{i+1},	&	1 \leq i \leq n \, . 
\end{array}
\end{cases}
\end{equation}

We are now in conditions to derive the existence and uniqueness result. 

\begin{prop}
For any initial condition $x_0 \in \Omega$, the sequence of Cauchy 
problems given by \eqref{eq:sequence-Cauchy-pb} admits a unique global 
solution, denoted again by $x(t,\,x_0)$, whose components are non-negative.
Furthermore, the region $\Omega$ is positively invariant.
\end{prop}

\begin{proof}
Applying Theorem~\ref{theo:well-posedness-model}, a finite number of times, 
directly provides the existence and uniqueness of global solutions 
to problem \eqref{eq:sequence-Cauchy-pb}.
\end{proof}

Note that the solutions of problem \eqref{eq:sequence-Cauchy-pb} are continuous 
on the time interval $[T_0,\,+\infty)$, but may not be of class $\mathscr{C}^1$ 
at $t = T_i$, $0 \leq i \leq n-1$. From the modeling point of view,
each change of parameters occurring at time $t = T_i$ ($1 \leq i \leq n-1$)
corresponds, for example, to a public announcement of confinement/lift 
of confinement or prohibition of displacement. 


\section{Existence of pseudo-periodic solutions}
\label{sec:exist:pseudoperiodic}

In this section, we show that piecewise constant parameters can lead 
to pseudo-periodic solutions. This result has important practical applications 
in the context of the COVID-19 pandemic. During the first months of the COVID-19 pandemic, 
the first concern of governments was to decrease the level of infected individuals. 
This was achieved, in great part, by confinement decisions. However, those confinement 
decisions had a large impact on the economy, and it was important to decide to lift this ban.
Consequently, the risk of a premature relaxation was the second concern of governments, 
since it was feared to provoke a second wave of infection. The existence of multiple epidemic waves 
in a pandemic can be mathematically justified by the following theorem. 

\begin{theo}
Assume that the disease-free equilibrium, $\Sigma_0$, admits a non-trivial 
basin of attraction $\Omega_0 \subset \Omega$ if $R_0 < 1$, and that the 
endemic equilibrium $\Sigma^+$ admits a non-trivial basin of attraction 
$\Omega^+ \subset \Omega$ if $R_0 > 1$. Let $\alpha_0$ and $\alpha^+$ denote 
two sets of parameters of system \eqref{eq:model} such that
$R_0(\alpha_0) < 1$ and $R_0(\alpha^+) > 1$.
Let $x_0 \in \Omega_0$ and consider the sequence of Cauchy problems
\begin{equation}
\label{eq:sequence-pseudo-periodic-solutions}
\begin{cases}
\begin{array}{lllll}
&x(T_0) = x_0, &\dot{x}(t) = f\big(x(t),\,\alpha_0\big),&T_0 < t < T_1,	\\
&x(T_i) = \lim\limits_{\substack{t \to T_i\\t \in (T_{i-1},\,T_i)}} x(t),
&\dot{x}(t) = f\big(x(t),\,\alpha^+\big), &T_i < t < T_{i+1},
& \text{for}~i~\text{odd},\\
&x(T_i) = \lim\limits_{\substack{t \to T_i\\
t \in (T_{i-1},\,T_i)}} x(t), &\dot{x}(t) = f\big(x(t),\,\alpha_0\big), 
&T_i < t < T_{i+1}, & \text{for}~i~\text{even},
\end{array}
\end{cases}
\end{equation}
for $ 1 \leq i \leq n$, where $T_i$ is such that $x(T_i) \in \Omega^+$, 
for $i$ odd, and $x(T_{i}) \in \Omega_0$, for $i$ even.
Then the solution $x(t,\,x_0)$, of the latter sequence of Cauchy problems, 
exhibits pseudo-oscillations between a neighborhood $\msN_0$ of $\Sigma_0$ 
and a neighborhood $\msN^+$ of $\Sigma^+$ in $\Omega$.
\end{theo}

\begin{proof}
The initial condition $x_0$ of the sequence of Cauchy problems 
\eqref{eq:sequence-pseudo-periodic-solutions} has been chosen in $\Omega_0$,
which is assumed to be the basin of attraction of $\Sigma_0$.
Since $\alpha_0(R_0) < 1$, then the solution $x(t,\,x_0)$ is attracted 
to $\Sigma_0$ and thus reaches a neighborhood $\msN_0$ of $\Sigma_0$
after a finite time $T_1$. Now it is assumed that $x(T_1)$ belongs 
to the basin of attraction $\Omega^+$ of $\Sigma^+$.
Since $\alpha^+(R_0) > 1$, then the solution $x(t,\,x_0)$ 
is now attracted to $\Sigma^+$ and thus reaches a neighborhood $\msN^+$ 
of $\Sigma^+$ after a finite time $T_2$. Next, we assume that $x(T_2)$ 
belongs to the basin of attraction $\Omega_0$ of $\Sigma_0$.
Since $\alpha_0(R_0) < 1$, then the solution $x(t,\,x_0)$ 
is now attracted to $\Sigma_0$ and thus reaches again the neighborhood 
$\msN_0$ of $\Sigma_0$ after a finite time $T_3$. Repeating those arguments 
a finite number of times leads to the desired conclusion, that is,
the solution $x(t,\,x_0)$ exhibits pseudo-oscillations between the neighborhoods 
$\msN_0$ and $\msN^+$ of $\Sigma_0$ and $\Sigma^+$, respectively.
\end{proof}

Note that the assumption $x(T_1) \in \Omega^+$ is directly satisfied 
for $T_1$ large enough, if $\Sigma_0$ belongs to the interior of $\Omega^+$ in $\Omega$.
Similarly, the assumption $x(T_2) \in \Omega_0$ is directly satisfied 
for $T_2$ large enough, if $\Sigma^+$ belongs to the interior of $\Omega_0$ in $\Omega$.
By virtue of Theorem~\ref{theo:glob:DFE}, the latter property is directly 
satisfied since $\Omega_0 = \Omega$. A numerical simulation of a solution 
of the sequence of Cauchy problems \eqref{eq:sequence-pseudo-periodic-solutions}, 
exhibiting pseudo-oscillations between a neighborhood $\msN_0$ of $\Sigma_0$ 
and a neighborhood $\msN^+$ of $\Sigma^+$ in $\Omega$, 
is depicted in Figure~\ref{fig:pseudo-periodic-solution}.
\begin{figure}
\centering
\includegraphics[width=0.85\textwidth]{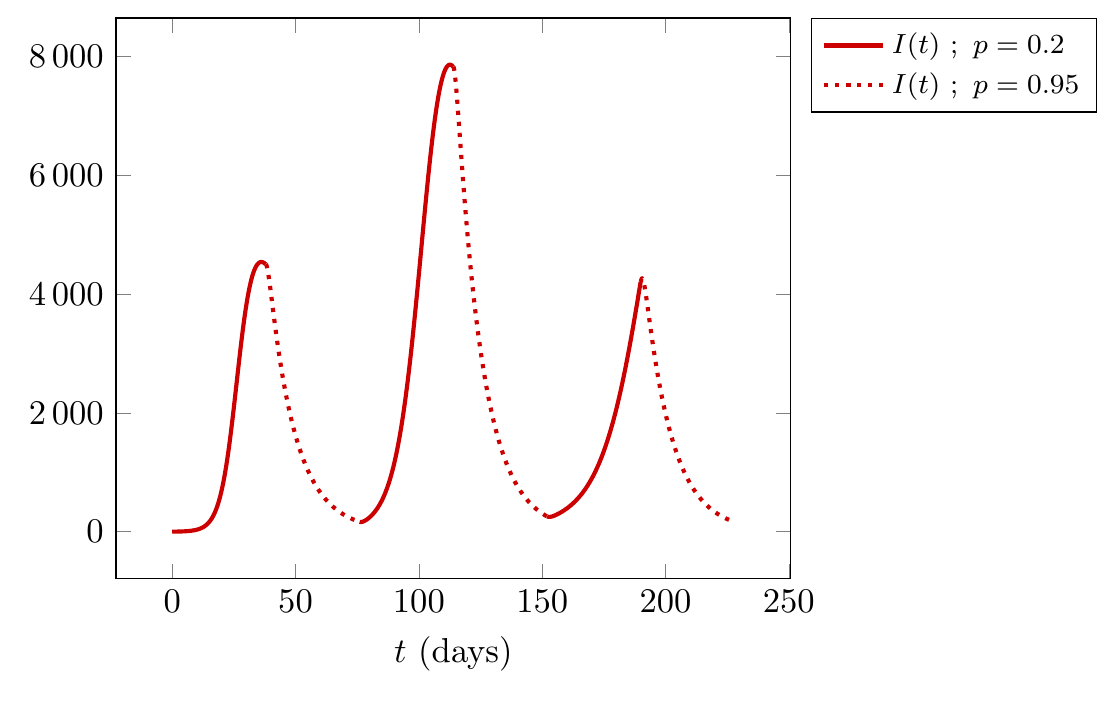}
\caption{Solution of the sequence of Cauchy problems
\eqref{eq:sequence-pseudo-periodic-solutions} exhibiting
pseudo-oscillations between a neighborhood $\msN_0$ of $\Sigma_0$ 
and a neighborhood $\msN^+$ of $\Sigma^+$ in $\Omega$. Here $I$ denotes 
the number of active infected individuals and $p$ the fraction, $0 < p < 1$, 
of susceptible individuals $S$ that is transferred to the protected class $P$.}
\label{fig:pseudo-periodic-solution}
\end{figure}


\section{Dynamics of a complex network of non-identical SAIRP models}
\label{sec:complexnet:model}

It is now widely admitted that mobilities play an important role
on the dynamics of epidemics, at numerous stages of their development.
In this section, we propose to study the propagation of the COVID-19 
outbreak in Portugal by modeling this country by a complex network
in which the six regions studied previously for the calibration 
of the model \eqref{eq:model} are considered.

In a first stage, we show how to construct a complex network of SAIRP models.
Let us consider the six regions of Portugal, as depicted in Figure~\ref{fig:portugal}.
Those six regions are connected by a finite number of links that define a graph 
$\msG = (\msV,\,\msE)$ made of a set $\msV$ of $6$ vertices, which correspond 
to the six regions (\emph{Norte}, \emph{Centro}, \emph{Lisboa e Vale do Tejo}, 
\emph{Alentejo}, \emph{Algarve}, \emph{Pinhal Litoral}), and of a set $\msE$ of edges, 
which model the main connections between those $6$ regions. In order to describe 
the state of each region, we couple each vertex of the graph with one instance
of the model \eqref{eq:sequence-Cauchy-pb}. Since each region has its own specificity,
we consider that the multiple instances of the model are non-identical, which means
that the values of the parameters can differ from one region to another.

We introduce the following notations:
\[
\begin{split}
&x_i = (S_i\,A_i,\,I_i,\,R_i,\,P_i)^T \in \R^5,\quad 1 \leq i \leq 6,\\
&X = (x_1,\,\dots,\,x_5)^T \in \left(\R^{5}\right)^6,\\
&H X = (H x_1,\,\dots,\,H x_6)^T \in \left(\R^{5}\right)^6,\\
&\alpha(t) = \big(\alpha_1(t),\,\dots,\,\alpha_6(t)\big) \in \left(\R^{9}\right)^6,
\end{split}
\]
where $H$ is the matrix of coupling strengths defined by
\begin{equation*}
H = 
\left[
\begin{array}{l l l l l}
\sigma_S &	0	       &	0			&	0			&	0		 \\
0		 &	\sigma_A   &	0			&	0			&	0		 \\
0		 &	0		   &	\sigma_I	&	0			&	0		 \\
0		 &	0		   &	0			&	\sigma_R	&	0		 \\
0		 &	0		   &	0			&	0			&	\sigma_P \\
\end{array}
\right],
\end{equation*}
with non negative coefficients $\sigma_S$, 
$\sigma_A$, $\sigma_I$, $\sigma_R$ and $\sigma_P$.

Next we define a matrix $L$ of connectivity as follows.
For each edge $(k,\,j) \in \mathscr{E}$, $k \neq j$,
we set $L_{j,k} = \eps_{j,k} > 0$.
If $(k,\,j) \notin \mathscr{E}$, $k \neq j$, we set $L_{j,k} = 0$.
The diagonal coefficients satisfy
\begin{equation*}
L_{j,j} = - \sum_{\substack{k = 1\\ k \neq j}}^n \eps_{k,j},
\end{equation*}
thus $L$ is a matrix whose sums of coefficients in each column is null.
For instance, the connectivity matrix of the graph corresponding 
to Figure~\ref{fig:portugal} is given by
\[
L = \left[
\begin{array}{cccccc}
L_{11}	  &	\eps_{12} &	0		  &	0			&	0			&	\eps_{16} \\
\eps_{21} &	L_{22}	  &	\eps_{23} &	\eps_{24}	&	0			&	\eps_{26} \\
0		  &	\eps_{32} &	L_{33}	  &	\eps_{34}	&	0			&	\eps_{36} \\
0		  &	\eps_{42} &	\eps_{43} &	L_{44}		&	\eps_{45}	&	0		  \\
0		  &	0		  &	0		  &	\eps_{54}	&	L_{55}		&	0	      \\
\eps_{61} &	\eps_{62} &	\eps_{63} &	0			&	0			&	L_{66}	  \\
\end{array}
\right],
\]
with
\[
\begin{split}
&L_{11} = -(\eps_{21}+\eps_{61}),						\\
&L_{22} = -(\eps_{12}+\eps_{32}+\eps_{42}+\eps_{62}),	\\
&L_{33} = -(\eps_{23}+\eps_{43}+\eps_{63}),				\\
&L_{44} = -(\eps_{24}+\eps_{34}+\eps_{54}),				\\
&L_{55} = -\eps_{45},									\\
&L_{66} = -(\eps_{16}+\eps_{26}+\eps_{36}).
\end{split}
\]
In this complex network model, we consider that an edge 
$(k,\,j) \in \mathscr{E}$, $k \neq j$, models a connection 
between two regions $k$ and $j$, which corresponds to human 
displacements from region $k$ towards region $j$. Moreover, 
the parameter $\sigma_S$ models the rate of susceptible 
individuals in region $k$ which migrate towards vertex $j$. 
The parameters $\sigma_A$, $\sigma_I$, $\sigma_R$ and $\sigma_P$ 
are defined analogously. This implies that our model can take 
into account the situation where a part of the population 
is not concerned with the migrations. For instance, it is 
relevant to consider $\sigma_I = \sigma_P = 0$,
while $\sigma_S > 0$ and $\sigma_A > 0$. The set of edges $\msE$ 
and the coupling strengths stored in the matrix $H$ define what is usually
called the \emph{topology} of the complex network.

Next, we explicit the equations that describe the state 
of region $j \in \lbrace 1,\,\dots,\,6 \rbrace$:
\begin{equation}
\label{eq:SAIRP-network}
\begin{cases}
\dot{S}_j = \Lambda_j - \beta_j (1-p_j)\frac{ \left( \theta_j A_j 
+ I_j \right)}{N_j} S_j - \phi_j p_j S_j + \omega_j P_j - \mu_j S_j
+ \sigma_S \displaystyle\sum_{k=1}^5 L_{j,k} S_k,\\[0.2 cm]
\dot{A}_j = \beta_j (1-p_j) \frac{\left( \theta_j A_j + I_j 
\right)}{N_j} S_j - \nu_j A_j - \mu_j A_j
+ \sigma_A \displaystyle\sum_{k=1}^5 L_{j,k} A_k,\\[0.2 cm]
\dot{I}_j = \nu_j A_j - \delta_j I_j - \mu_j I_j
+ \sigma_I \displaystyle\sum_{k=1}^5 L_{j,k} I_k,\\[0.2 cm]
\dot{R}_j = \delta_j I_j - \mu_j R_j
+ \sigma_R \displaystyle\sum_{k=1}^5 L_{j,k} R_k,\\[0.2 cm]
\dot{P}_j = \phi_j p_j S_j - \omega_j P_j - \mu_j P_j
+ \sigma_P \displaystyle\sum_{k=1}^5 L_{j,k} P_k,
\end{cases}
\end{equation}
where the time dependence is omitted, in order to lighten the notations.
In particular, the parameters can be determined by piecewise constant 
functions as in model \eqref{eq:sequence-Cauchy-pb}.

In the next theorem, we establish the existence and uniqueness 
of global solutions to the complex network problem \eqref{eq:SAIRP-network}
as for models \eqref{eq:model} and \eqref{eq:sequence-Cauchy-pb}.
Following \cite{cantin2019influence}, we introduce the minimum mortality 
rate $\mu_0$ defined by
\[
\mu_0 = \min\limits_{1 \leq j \leq 6} \mu_j,
\]
the positive coefficient $\Lambda_0$ defined by
\[
\Lambda_0 = \displaystyle\sum_{j=1}^6 \Lambda_j,
\]
and the compact region
\begin{equation}
\Theta = \left\lbrace (x_j)_{1\leq j \leq 30} \in (\R^+)^{30}~;~
\displaystyle\sum_{j=1}^{30} x_j \leq \frac{\Lambda_0}{\mu_0} \right\rbrace.
\label{eq:region-Theta}
\end{equation}

\begin{theo}
For any $X_0 \in \Theta$, the Cauchy problem given by \eqref{eq:SAIRP-network} 
and $X(0) = X_0$ admits a unique solution denoted by $X(t,\,X_0)$,
defined on $[0,\,\infty)$, whose components are non-negative.
Furthermore, the region $\Theta$ defined by \eqref{eq:region-Theta} is positively invariant.
\end{theo}

\begin{proof}
The existence and uniqueness of local in time solutions is immediate.
The non-negativity property is guaranteed by the quasi-positivity 
of the non-linear operator determined by the right hand side 
of system \eqref{eq:SAIRP-network}.
The total population in the complex network, defined by
\[
N(t) = \displaystyle\sum_{j=1}^6 \big[S_j(t) + A_j(t) 
+ I_j(t) + R_j(t) + P_j(t) \big], \quad t \geq 0,
\]
satisfies
\[
\dot{N}(t) \leq - \mu_0 N(t) + \Lambda_0,
\quad t \geq 0,
\]
since the matrix of connectivity $L$ is a zero column sum matrix.
It follows that
\[
N(t) \leq  \left[N(0) - \frac{\Lambda_0}{\mu_0}\right]e^{-\mu_0 t} 
+ \frac{\Lambda_0}{\mu_0}, \quad t \in [0,\,T],
\]
which leads to the intended conclusion.
\end{proof}

In the next section, we apply previous complex network to the real data 
from COVID-19 in Portugal, since the first confirmed active COVID-19 case, 
in March 2, 2020, until September 17, 2020. 


\section{Portugal case study: complex network with 6 regions}
\label{sec:casestudyPT}

In this section, we calibrate the number of active infected individuals 
$I$ given by the $SAIRP$ model with piecewise constant parameters 
\eqref{eq:sequence-Cauchy-pb}, to six distinct regions from Portugal. 
We show that the model allows to fit well the Portuguese real data available 
in \cite{DGS:covid19:PT}. Taking into account these parameter values, 
we perform numerical simulations of the complex network problem \eqref{eq:SAIRP-network}, 
where the main goal is to investigate the effect of the topology on the dynamics 
of the epidemics, and determine a topology that minimizes 
the level of active infected individuals. 


\subsection{Model with piecewise constant parameters: fitting to real data from 6 Portuguese regions}
\label{sec:fit-the-model}

We first consider the $SAIRP$ model with piecewise constant parameters 
\eqref{eq:sequence-Cauchy-pb} and show that the class $I$ fits well 
the active confirmed cases of infected individuals provided by the 
Portuguese National Authorities \cite{DGS:covid19:PT} 
for six regions of Portugal mainland, as depicted in Figure~\ref{fig:portugal}, 
where (1) represents the region \emph{Norte}, (2) \emph{Centro}, 
(3) \emph{Lisboa e Vale do Tejo}, (4) \emph{Alentejo}, 
(5) \emph{Algarve} and (6) \emph{Pinhal Litoral}. 
\begin{figure}[ht!]
\centering
\includegraphics[scale=0.99]{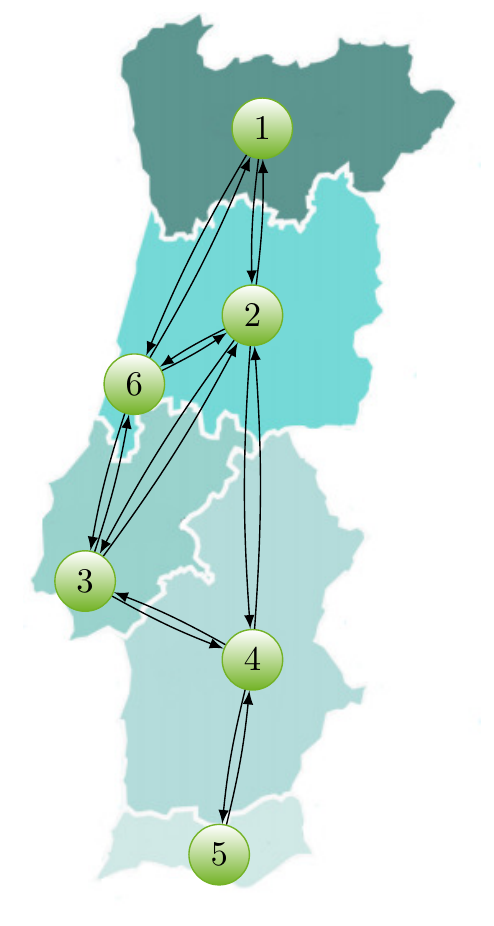}
\caption{Six regions of Portugal and some of their connections: 
(1) \emph{Norte}; (2) \emph{Centro}; (3) \emph{Lisboa e Vale do Tejo}; 
(4) \emph{Alentejo}; (5) \emph{Algarve}; (6) \emph{Pinhal litoral}.}
\label{fig:portugal}
\end{figure}

The initial conditions $S_0$, $A_0$, $I_0$, $R_0$ and $P_0$ at time 
$t = 1$, where $t=1$ corresponds to the day of the first confirmed 
cases of COVID-19 disease in each region, which differ from one region to another, 
see Table~\ref{table:init:regions}. Namely, the first cases occurred 
in each region at: \emph{Norte}, March 2, 2020; \emph{Centro},  March 3, 2020; 
\emph{Lisboa e Vale do Tejo}, March 3, 2020; \emph{Alentejo}, March 18, 2020; 
\emph{Algarve}, March 8, 2020.  
\begin{table}[!htb]
\caption{Initial conditions for each of the six regions in Portugal.}
\label{table:init:regions}	
\centering
\begin{tabular}[center]{ l l l l l l }  \hline
Region  & $S_0$ & $A_0$ & $I_0$ & $R_0$ & $P_0$  \\  \hline 
\emph{Norte} & $3125804$ & $2$ & $\frac{2}{0.15}$ & $0$ & $0$\\ 
\emph{Centro} & $1480664$  & $1$ & $\frac{1}{0.15}$  & $0$ & $0$ \\
\emph{Lisboa e Vale do Tejo} &  $3659871$ &  $1$ & $\frac{1}{0.15}$  & $0$ & $0$ \\
\emph{Alentejo} & $166726$  & $2$ & $\frac{2}{0.15}$  & $0$ & $0$ \\
\emph{Algarve} & $451006$  & $1$ & $\frac{1}{0.15}$  & $0$ & $0$ \\
\emph{Pinhal litoral} & $271078$ & $3$ &  $\frac{3}{0.15}$ & $0$ & $0$ \\ \hline	
\end{tabular}  
\end{table}

Some of the parameters take constant values for all the period of time 
under consideration and are the same for the six regions, 
see Table~\ref{table:param:constant}. Others, even taking constant values 
for all time window, differ from one region to another, 
see Table~\ref{table:param:const:difer}.
\begin{table}[!htb]
\caption{Constant parameter values that take the same values for the six regions, in Portugal.}
\label{table:param:constant}	
\centering
\begin{tabular}[center]{ l l l l }  \hline
Parameter & Description & Value & Reference \\  \hline 
$\Lambda$ & Recruitment rate & $\frac{0.19\%\times N_0}{365}$ & \cite{pordata} \\ 
$\mu$ & Natural death rate & $\frac{1}{81 \times 365}$ & \cite{pordata} \\
$\theta$ & Modification parameter & $1$ & \cite{Silvaetal:covidPT} \\
$v$  & Transfer rate from $A$ to $I$ & $1$ & \cite{Silvaetal:covidPT} \\
$q$  & Fraction of $A$ individuals that are confirmed to be infected & $0.15$ & \cite{Silvaetal:covidPT}\\ 
\hline	
\end{tabular}  
\end{table}
\begin{table}[!htb]
\caption{Constant parameter values that differ from one region to another, in Portugal.}
\label{table:param:const:difer}	
\centering
\begin{tabular}[center]{ l c c c } \hline
Region  & $\phi$ & $\delta$ & $w$ \\  
& (\footnotesize transfer rate from $S$ to $P$) 
& (\footnotesize transfer rate from $I$ to $R$) 
&  (\footnotesize transfer rate from $P$ to $S$) \\ \hline 
\emph{Norte} &  $\phi=\frac{1}{12}$ & $\delta=\frac{1}{27}$ & $w=\frac{1}{45}$ \\ 
\emph{Centro} &  $\phi=\frac{1}{11}$ & $\delta=\frac{1}{27}$ & $w=\frac{1}{45}$\\ 
\emph{\small Lisboa e Vale do Tejo} &  $\phi=\frac{1}{11}$ & $\delta=\frac{1}{27}$ & $w=\frac{1}{45}$ \\
\emph{Alentejo} & $\phi=1$ & $\delta=\frac{1}{21}$ & $w=\frac{1}{41}$  \\
\emph{Algarve} &  $\phi=\frac{1}{6}$ & $\delta=\frac{1}{21}$ & $w=\frac{1}{45}$  \\
\emph{Pinhal litoral} &  $\phi=1$ & $\delta=\frac{1}{21}$ & $w=\frac{1}{43}$ \\ \hline	
\end{tabular}  
\end{table}

The parameters $\beta$, $m$ and $p$ are piecewise continuous, taking 
different values depending on the considered time interval and the region, 
see Tables~\ref{table:param:piecewise:1} and \ref{table:param:piecewise:2}.
\begin{table}[!htb]
\caption{Piecewise constant parameter values for the regions 
\emph{Norte}, \emph{Centro} and \emph{Lisboa e Vale do Tejo}}
\label{table:param:piecewise:1}	
\centering
\begin{tabular}[center]{ l c c c }  \hline
Region  & $\beta$ & $p$  & $m$  \\  
& (\footnotesize transmission rate) & (\footnotesize transfer fraction from $S$ to $P$) 
&  (\footnotesize transfer fraction from $P$ to $S$) \\ \hline 
\emph{\bf Norte} &  &   &  \\
Time interval: $[1;75]$ & $\beta_1=1.40$ & $m_1=0.09$ & $p_1=0.675$\\  
Time interval: $[75;122]$ & $\beta_2=0.15$ & $m_2=0.15$ & $p_2=0.60$ \\
Time interval: $[122;170]$ & $\beta_3=1.28$ & $m_3=0.14$ & $p_3=0.56$ \\
Time interval: $[170;200]$ & $\beta_4=1.46$ & $m_4=0.17$ & $p_4=0.52$\\ \hline
\emph{\bf Centro} &   &   &  \\ 
Time interval: $[1;69]$ & $\beta_1=1.351$ & $m_1=0.10$ & $p_1=0.675$ \\  
Time interval: $[69;94]$ &	 $\beta_2=0.45$ & $m_2=0.09$ & $p_2=0.60$ \\
Time interval: $[94;164]$ & $\beta_3=1.10$ & $m_3=0.16$ & $p_3=0.56$\\
Time interval: $[164;199]$ & $\beta_4=1.33$ & $m_4=0.17$ & $p_4=0.54$\\ \hline
\emph{\small \bf Lisboa e Vale do Tejo} &   &   &  \\
Time interval: $[1;36]$ &	$\beta_1=1.45$ & $m_1=0.16$ & $p_1=0.675$\\
Time interval: $[36;54]$ & $\beta_2=1.13$ & $m_2=0.13$ & $p_2=0.69$\\
Time interval: $[54;129]$ & $\beta_3=1.41$ & $m_3=0.13$ & $p_3=0.56$\\	
Time interval: $[129;160]$ & $\beta_4=1.10$ & $m_4=0.11$ & $p_4=0.57$\\
Time interval: $[160;199]$ & $\beta_5=1.43$ & $m_5=0.17$ & $p_5=0.54$ \\ \hline
\end{tabular}  
\end{table}
\begin{table}[!htb]
\caption{Piecewise constant parameter values for the regions \emph{Alentejo}, 
\emph{Algarve} and \emph{Pinhal Litoral}}
\label{table:param:piecewise:2}	
\centering
\begin{tabular}[center]{ l c c c } \hline
Region  & $\beta$ & $p$  & $m$  \\  
& (\footnotesize transmission rate) & (\footnotesize transfer fraction from $S$ to $P$) 
&  (\footnotesize transfer fraction from $P$ to $S$) \\ \hline
\emph{\bf Alentejo} &  &  &   \\
Time interval: $[1;11]$ & $\beta_1=3.50$ & $m_1=0.40$ & $p_1=0.675$\\
Time interval: $[11;29]$ & $\beta_2=7.80$ & $m_2=0.45$ & $p_2=0.45$\\
Time interval: $[29;84]$ & $\beta_3=4.40$ & $m_3=0.38$ & $p_3=0.55$\\
Time interval: $[84;104]$ & $\beta_4=6.90$ & $m_4=0.46$ & $p_4=0.48$\\
Time interval: $[104;124]$ & $\beta_5=2.70$ & $m_5=0.30$ & $p_5=0.56$\\
Time interval: $[124;184]$ & $\beta_6=4.10$ & $m_6=0.36$ & $p_6=0.44$\\ \hline
\emph{\bf Algarve} &  &  &  \\
Time interval: $[1;36]$ & $\beta_1=1.11$ & $m_1=0.10$ & $p_1=0.60$\\
Time interval: $[36;69]$ & $\beta_2=1.37$ & $m_2=0.20$ & $p_2=0.52$\\
Time interval: $[69;91]$ & $\beta_3=0.80$ & $m_3=0.20$ & $p_3=0.55$\\
Time interval: $[91;109]$ & $\beta_4=1.90$ & $m_4=0.32$ & $p_4=0.45$\\
Time interval: $[109;149]$ & $\beta_5=1.35$ & $m_5=0.22$ & $p_5=0.58$\\
Time interval: $[149;194]$ & $\beta_6=1.35$ & $m_6=0.30$ & $p_6=0.55$\\ \hline
\emph{\bf Pinhal litoral} &   &  &  \\
Time interval: $[1;7]$ & $\beta_1=3.00$ & $m_1=0.09$ & $p_1=0.675$\\
Time interval: $[7;23]$ & $\beta_2=8.15$ & $m_2=0.25$ & $p_2=0.44$\\
Time interval: $[23;77]$ & $\beta_3=5.00$ & $m_3=0.16$ & $p_3=0.55$\\
Time interval: $[77;115]$ & $\beta_4=7.80$ & $m_4=0.25$ & $p_4=0.42$\\
Time interval: $[115;145]$ & $\beta_5=4.15$ & $m_5=0.15$ & $p_5=0.52$\\
Time interval: $[145;186]$ & $\beta_6=7.75$ & $m_6=0.27$ & $p_6=0.42$\\ \hline	
\end{tabular}  
\end{table}

Figures~\ref{fig:fit:Norte:Centro}--\ref{fig:fit:algarve:pinhal} show that 
the $SAIRP$ model \eqref{eq:sequence-Cauchy-pb} with piecewise constant parameters 
describes well the active infected cases in the six Portuguese regions under consideration.  
\begin{figure}[ht!]
\centering
\includegraphics[scale=0.5]{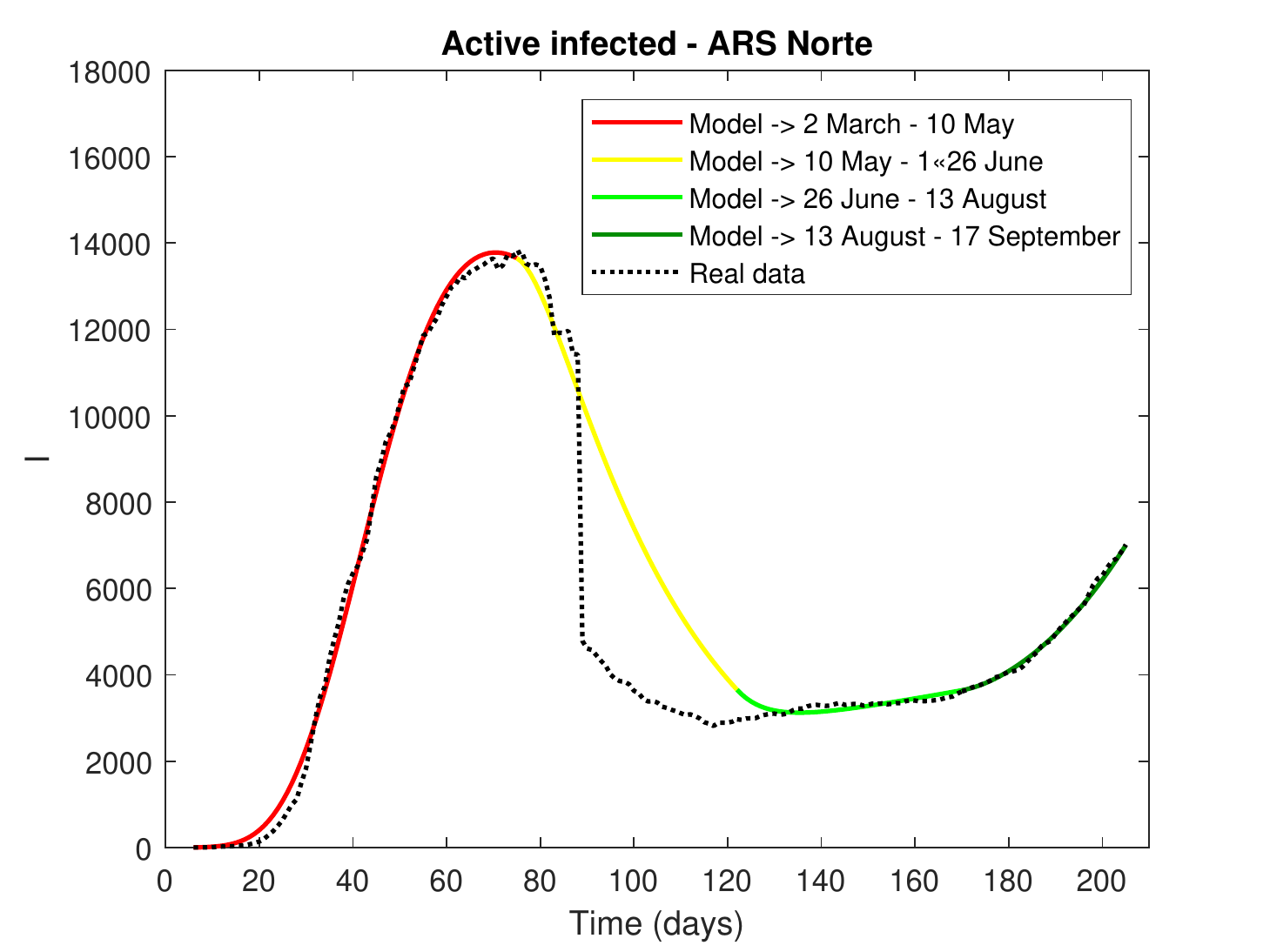}
\includegraphics[scale=0.5]{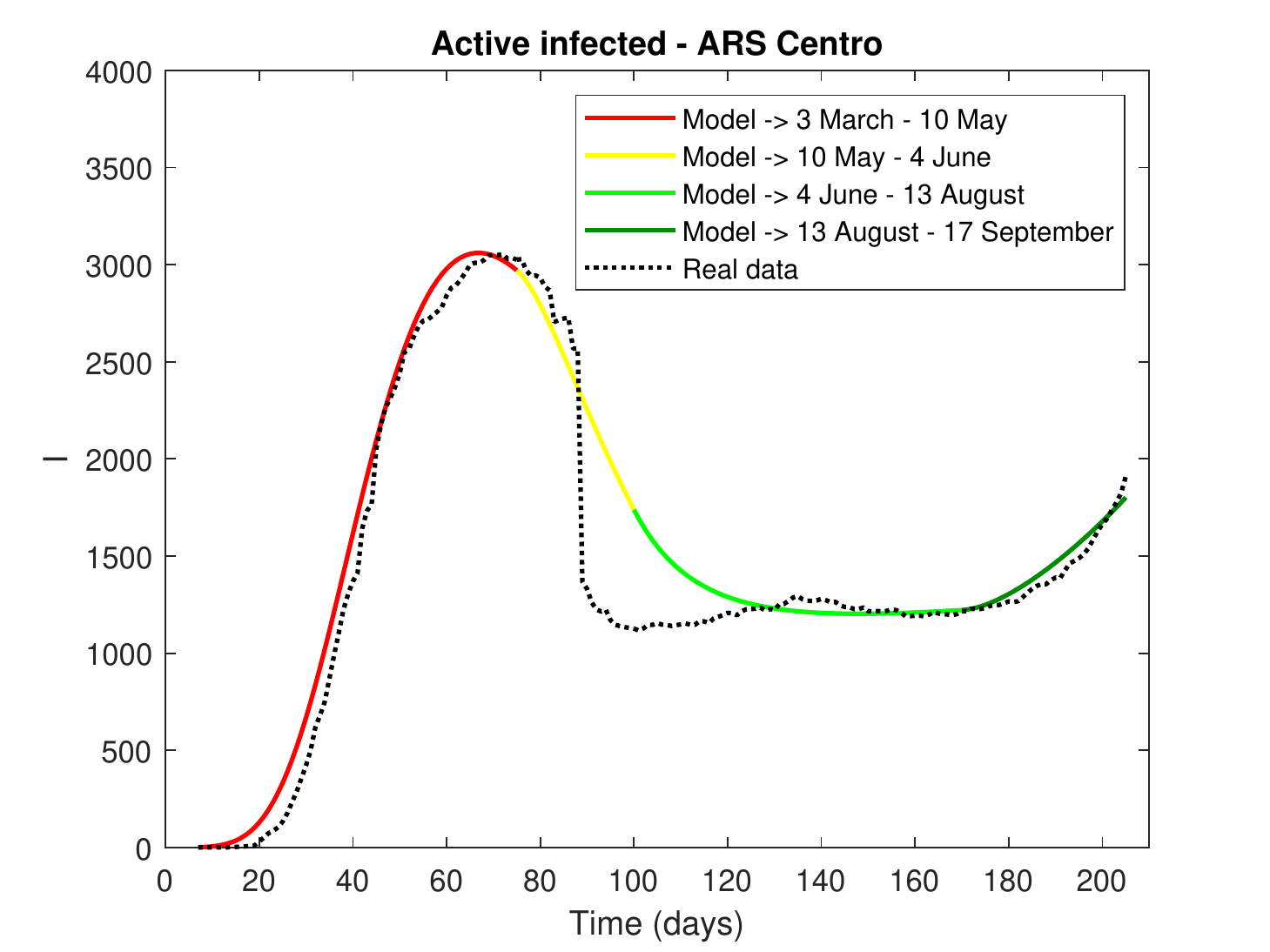}
\caption{Active infected cases of COVID-19. The black dotted 
lines represent the real data and the continuous lines represent 
the active infected individuals, $I$, output of the $SAIRP$ model 
with piecewise constant parameters from Table~\ref{table:param:piecewise:1}. 
Left: region \emph{Norte}.  Right: region \emph{Centro}.}
\label{fig:fit:Norte:Centro}
\end{figure}
\begin{figure}[ht!]
\centering
\includegraphics[scale=0.5]{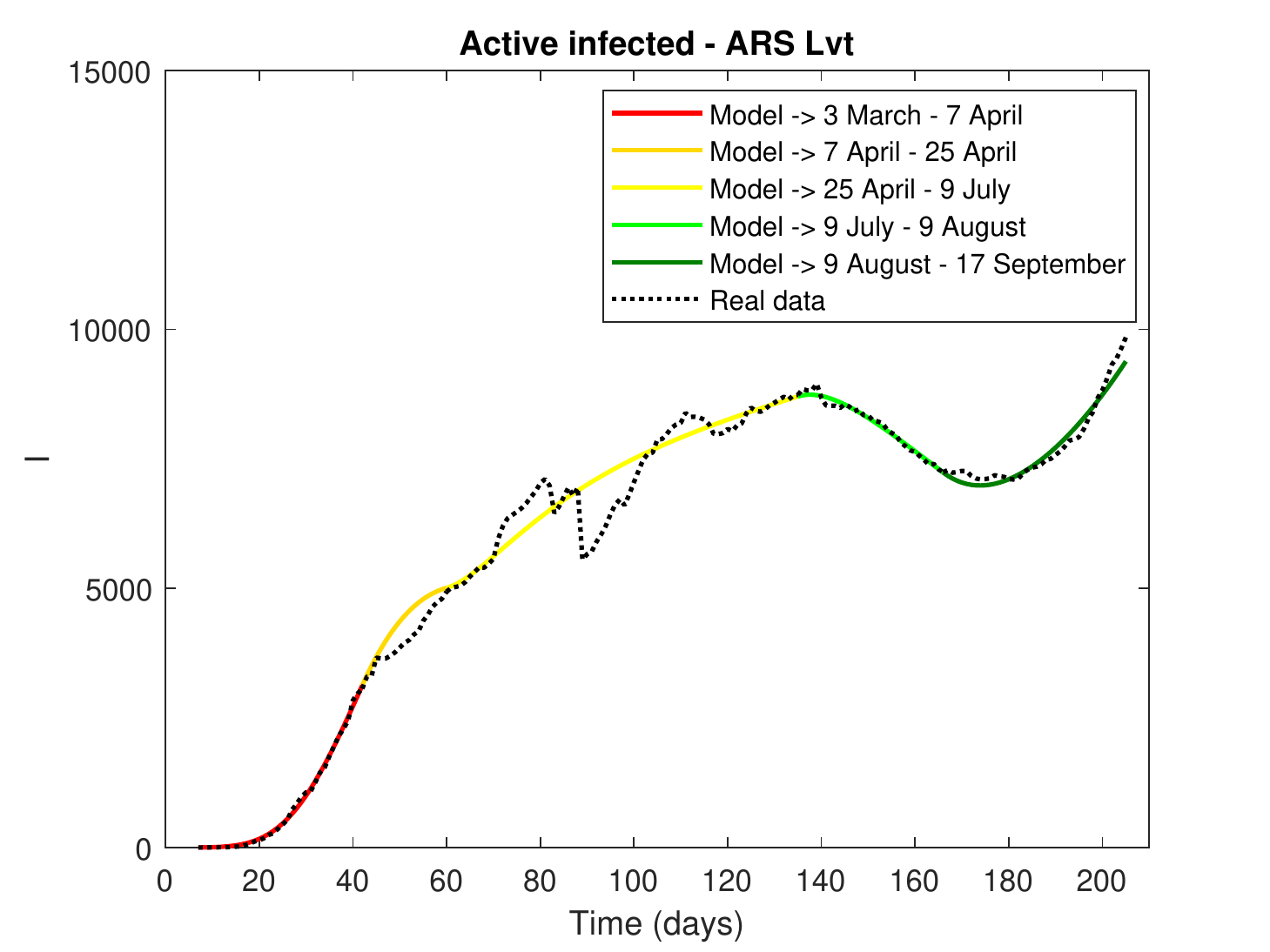}
\includegraphics[scale=0.5]{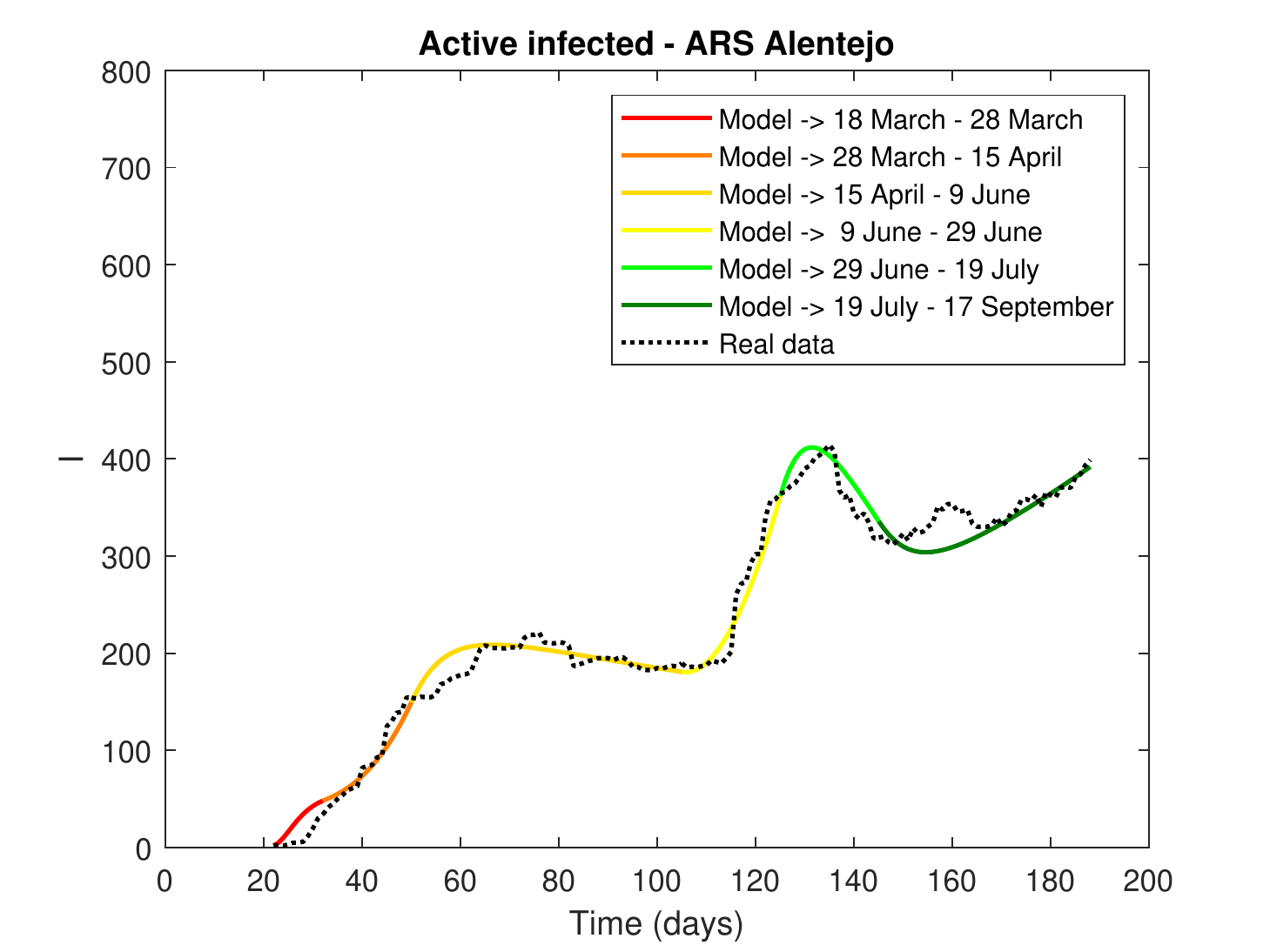}
\caption{Active infected cases of COVID-19. The black dotted lines 
represent the real data and the continuous lines represent the active 
infected individuals, $I$, output of the $SAIRP$ model with piecewise 
constant parameters from Tables~\ref{table:param:piecewise:1}--\ref{table:param:piecewise:2}. 
Left: region \emph{Lisboa e Vale do Tejo}. Right: region \emph{Alentejo}.}
\label{fig:fit:LVT:Alent}
\end{figure}
\begin{figure}[ht!]
\centering
\includegraphics[scale=0.5]{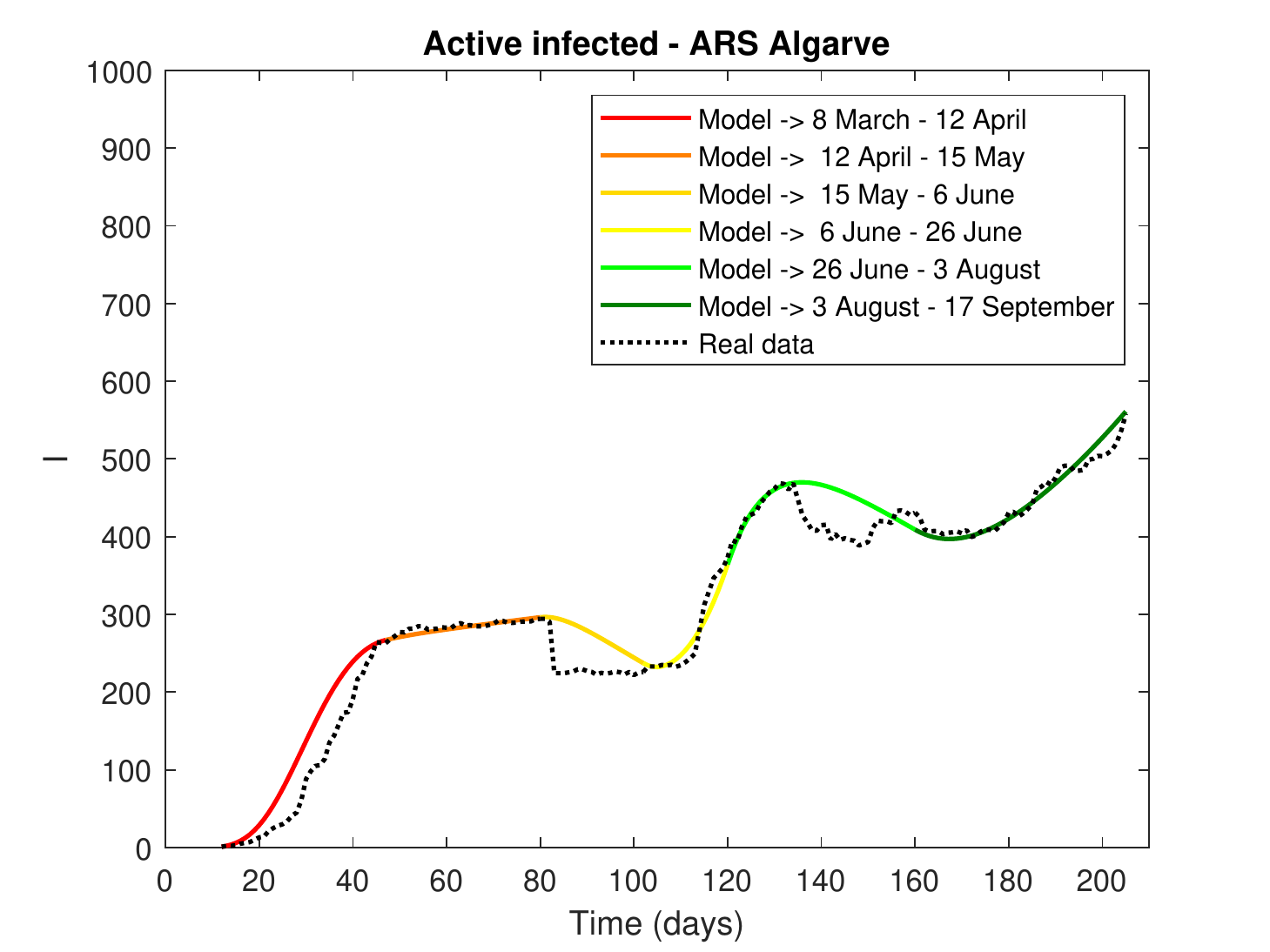}
\includegraphics[scale=0.5]{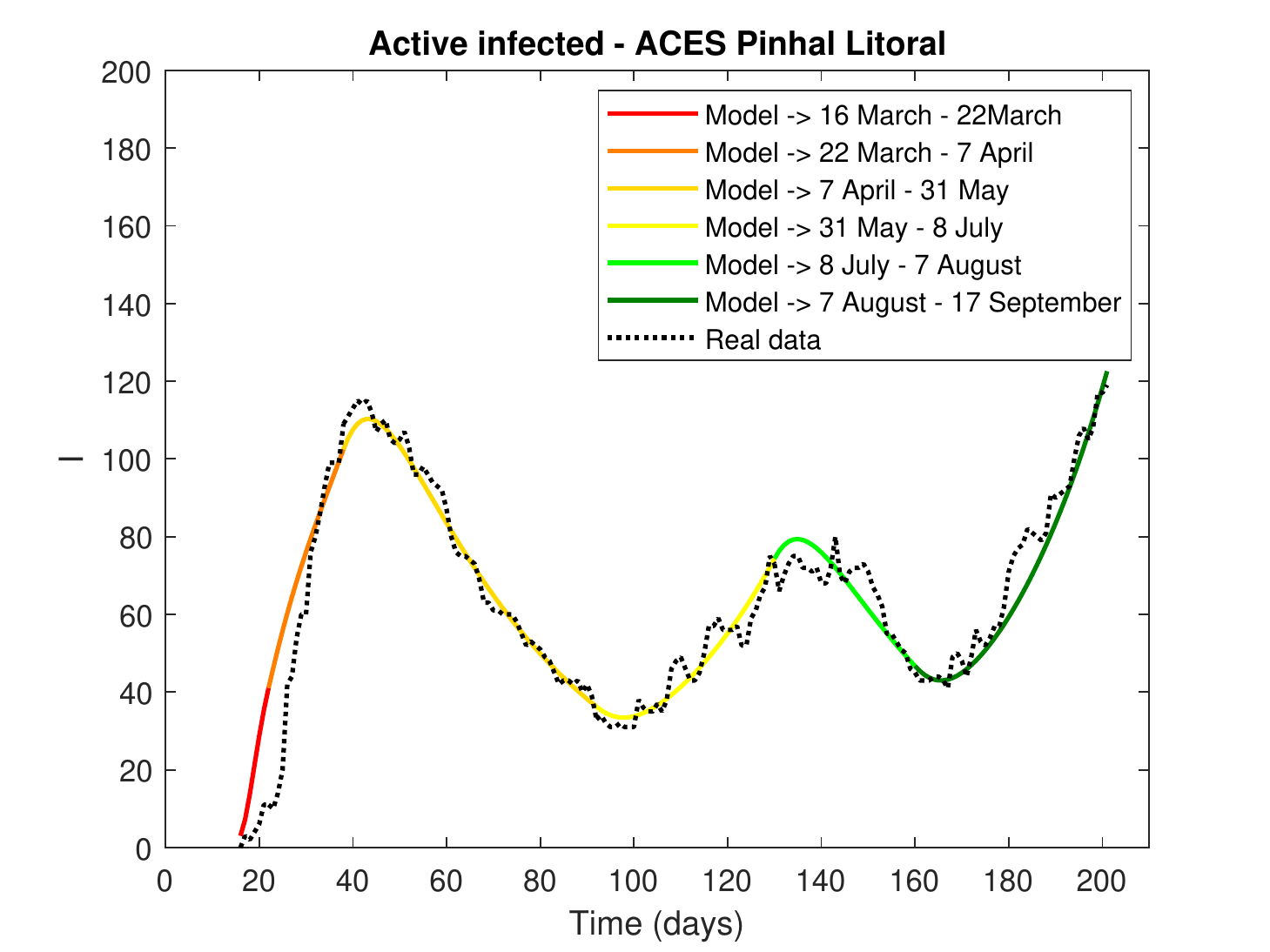}
\caption{Active infected cases of COVID-19. The black dotted lines 
represent the real data and the continuous lines represent the active 
infected individuals, $I$, output of the $SAIRP$ model with piecewise 
constant parameters from Table~\ref{table:param:piecewise:2}. 
Left: region \emph{Algarve}. Right: region \emph{Pinhal litoral}.}
\label{fig:fit:algarve:pinhal}
\end{figure}

Moreover, the active infected cases of all Portugal can also be modeled, 
see Figure~\ref{fig:fit:portugal}. 
\begin{figure}[ht!]
\centering
\includegraphics[scale=0.7]{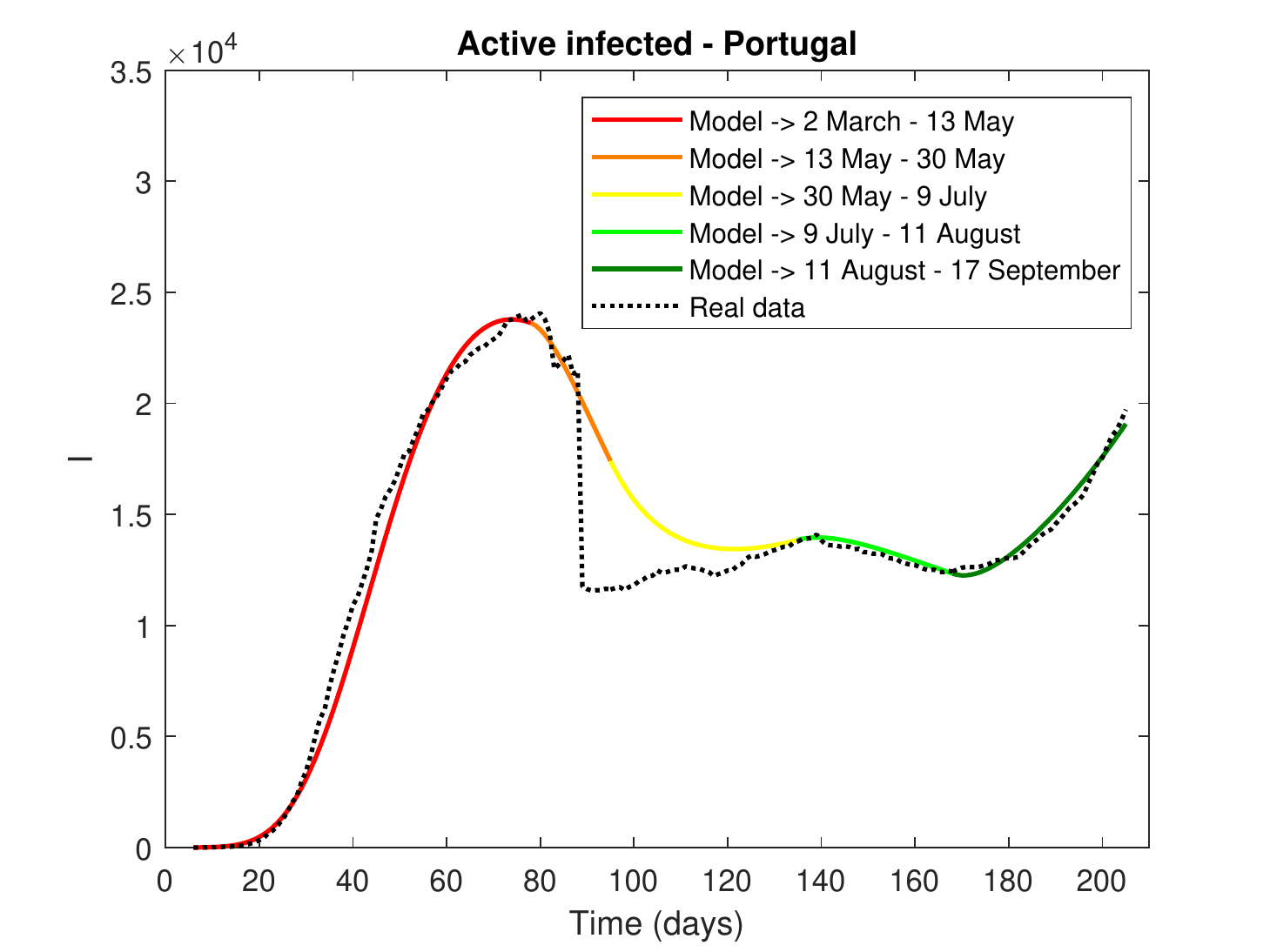}
\caption{Active infected cases of COVID-19 in Portugal. The black 
dotted lines represent the real data and the continuous lines 
represent the active infected individuals, $I$, output of the $SAIRP$ 
model with piecewise constant parameters. Initial conditions: 
$S_0 =  10295894$, $I_0 = 2$, $A_0 = \tfrac{2}{0.15}$, $R_0 = 0$, $Q_0 = 0$. 
The constant parameters take the values $\phi=1$, $\delta=\frac{1}{27}$, 
$w=\frac{1}{45}$ and the remaining constant parameters take the values 
in Table~\ref{table:param:constant}. The first infected case with 
COVID-19 in Potugal, $t=1$, occurred on March 2, 2020. 
Piecewise constant parameters:  $t \in [1;73]$,				
$\beta_1=1.505; m_1=0.09; p_1=0.675$;
$t \in [73;90]$, $\beta_2=0.50; m_2=0.09; p_2=0.65$;
$t \in [90;130]$,
$\beta_3=1.15; m_3=0.18; p_3=0.58$;
$t \in [130;163]$,
$\beta_4=0.96; m_4=0.16; p_4=0.61$;
$t \in [163;200]$,
$\beta_5=1.50; m_5=0.17; p_5=0.58$. }
\label{fig:fit:portugal}
\end{figure}


\subsection{Complex network model: numerical simulations for COVID-19 in Portugal}

Considering the six regions from Figure~\ref{fig:portugal}, and the parameter values 
from Section~\ref{sec:fit-the-model}, we perform numerical simulations of the complex 
network problem \eqref{eq:SAIRP-network}. The main goal is to investigate the effect 
of the topology on the dynamics of the epidemics. In particular, we analyze the existence 
of a topology which minimizes the average number of active infected individuals during 
a fixed time interval. Moreover, we analyze if other topologies are likely to worsen 
the level of infection.

In order to model the mobilities of susceptible and asymptomatic individuals, 
we set $\sigma_S = \sigma_A > 0$, whereas we fix $\sigma_I = \sigma_R = \sigma_P = 0$.
We test a sample of $1000$ randomly generated topologies among $2^{16}$ topologies 
(\emph{id est} sets of edges), for $\sigma_S = \sigma_A \in [0.01,\,0.1]$.
For each randomly generated topology, we perform a numerical integration 
of the complex network problem \eqref{eq:SAIRP-network}, with the same parameters 
$\alpha_1,\,\alpha_2,\,\dots,\,\alpha_6$ as considered in Section~\ref{sec:fit-the-model}. 
The computation has been performed with \texttt{Python 3.5} language, 
in a Debian/Gnu-Linux environment. For each integration, we have computed 
the number of average number of active infected individuals, that is, 
individuals of the class $I$, per day. The results are depicted 
in Figure~\ref{fig:sample-random-topologies}, where the average number 
of active infected individuals, per day, is depicted for the randomly 
generated topologies and $\sigma_S = \sigma_A = 0.01$,
and in Figure~\ref{fig:number-infected-topology-wrt-topology},
where the average number of active infected individuals is depicted 
with respect to the coupling strengths $\sigma_S = \sigma_A$
for a couple of remarkable topologies.
\begin{figure}[ht!]
\centering
\includegraphics[scale=0.8]{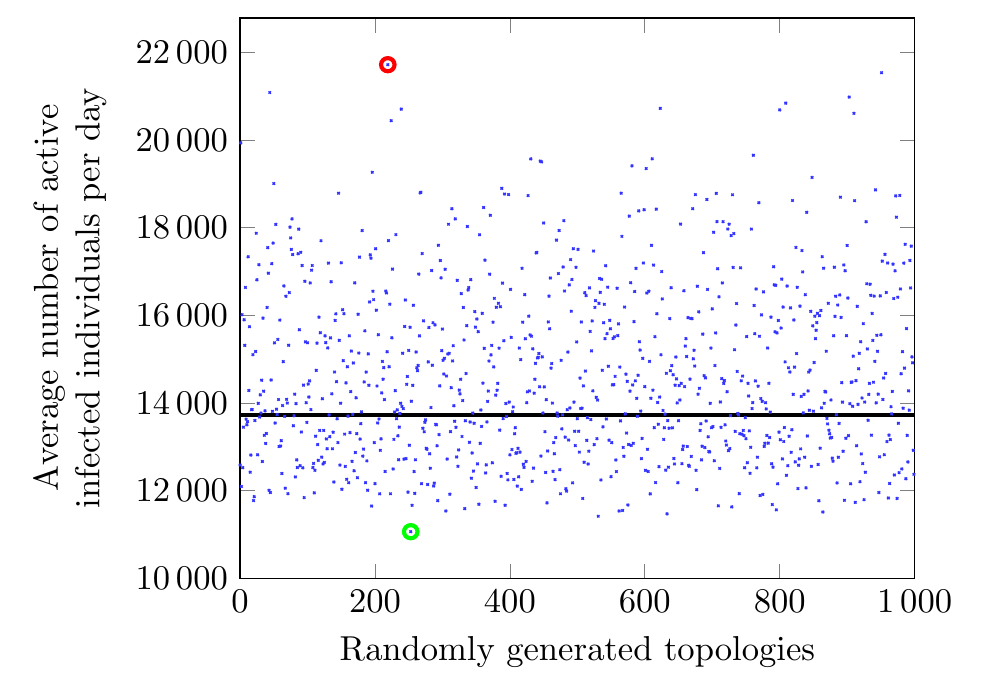}
\caption{Average number of active infected individuals, per day, f
or a sample of $1000$ randomly generated topologies, for $\sigma_S = \sigma_A = 0.01$.
The black line shows the level of infection for the empty topology.
The green circle shows the optimum topology which minimizes the level of infection,
whereas the red circle shows the topology which leads to the highest level of infection.}
\label{fig:sample-random-topologies}
\end{figure}
\begin{figure}[ht!]
\centering
\includegraphics[scale=0.57]{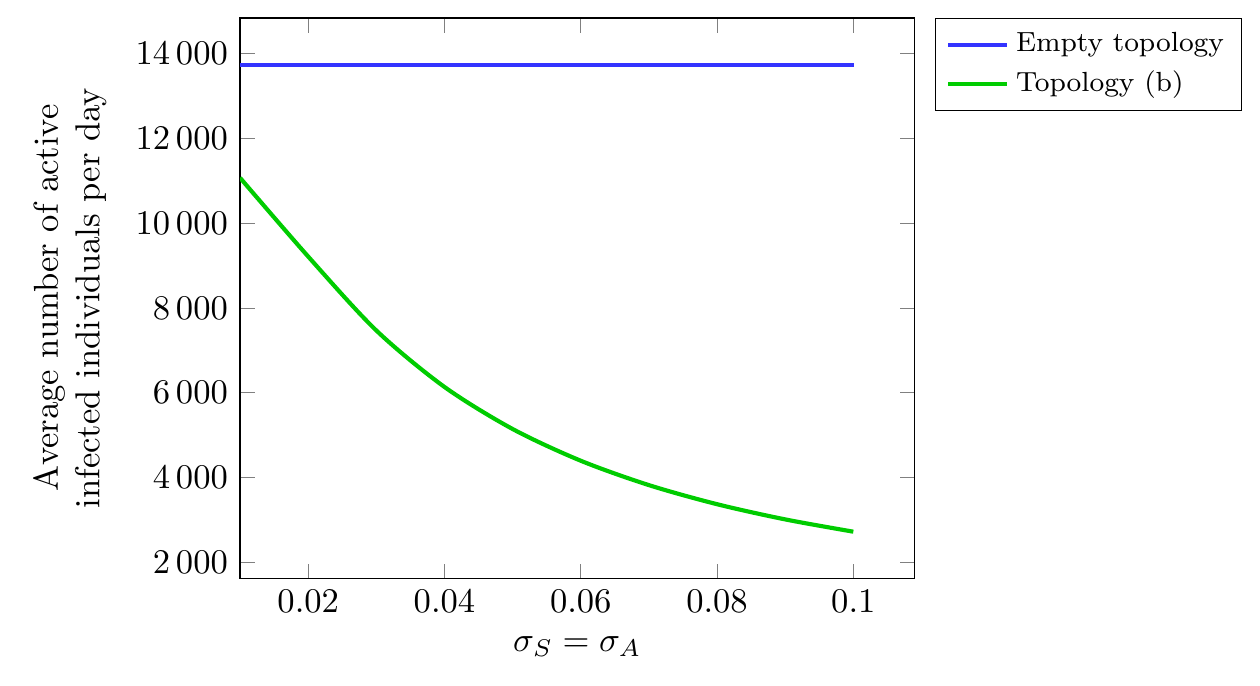}
\includegraphics[scale=0.57]{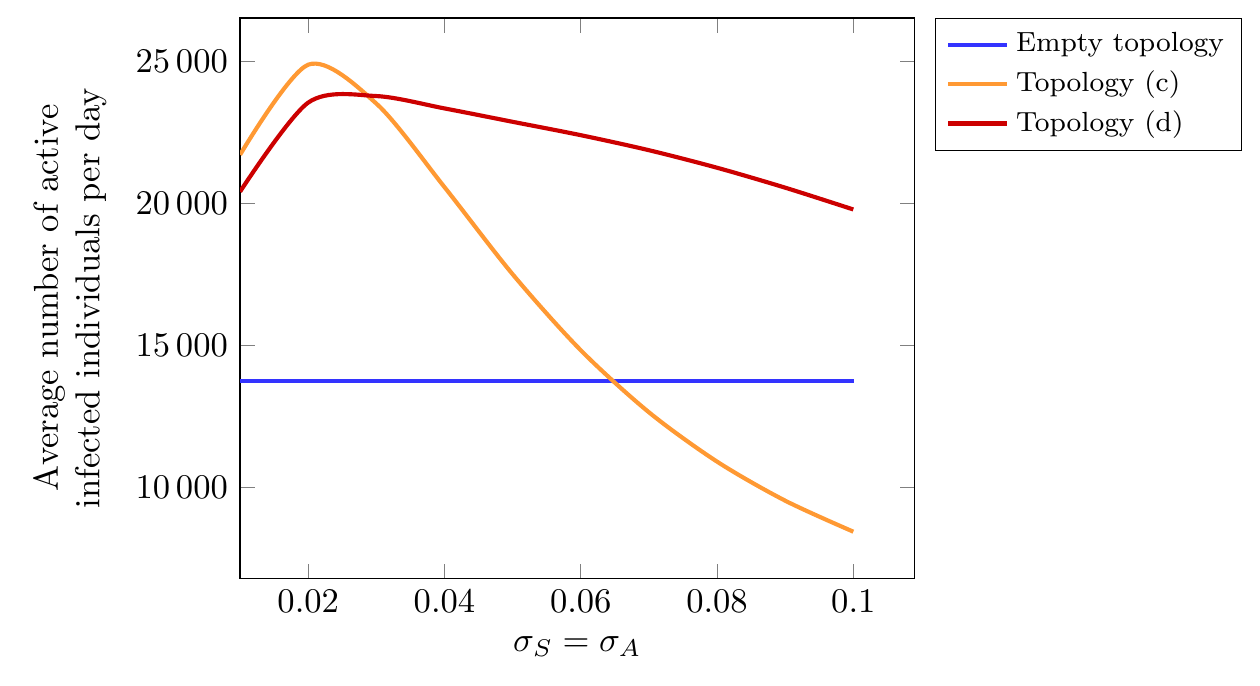}
\caption{Average number of active infected individuals, for each topology, 
with respect to the coupling strengths $\sigma_S = \sigma_A$.
Left: optimum topology which minimizes the level of infection of the epidemics.
Right: two examples of topologies that can increase the level of infection, 
compared to the empty topology, which corresponds to the situation where individuals 
do not migrate from one region to another. Topology (c) leads to a level of infection 
that overcomes the level of the empty topology for only a weak coupling strength,
whereas topology (d) seems to permanently overcome the level of the empty topology.}
\label{fig:number-infected-topology-wrt-topology}
\end{figure}
The numerical results reveal the existence of a certain number of topologies 
that decrease the level of infection, compared to the empty topology, 
which corresponds to the situation where individuals do not migrate from one region to another.
Among those topologies, one minimizes the level of cumulated infected individuals.
This optimal topology is depicted in Figure~\ref{fig:four-topologies}~(b). 
It is worth noting that this optimal topology connects the regions 
\emph{Norte}, \emph{Centro}, \emph{Lisboa e Vale do Tejo} and \emph{Pinhal Litoral}, 
but it does not connect those four regions to \emph{Alentejo} and \emph{Algarve}.
In parallel, the numerical simulations exhibit numerous topologies that increase 
the level of infection. We have presented in Figure~\ref{fig:four-topologies}~(c) 
and (d) two examples of such topologies. The corresponding levels of infection 
are presented in Figure~\ref{fig:number-infected-topology-wrt-topology}.
Topology (c) leads to a level of infection that overcomes the level of the empty topology 
for only a weak coupling strength, whereas topology (d) seems to permanently overcome 
the level of the empty topology. We remark that the topologies that increase the level 
of infection seem to connect the 6 regions of Portugal to the region \emph{Alentejo}. 
On the other hand, it can be helpful to authorize parsimonious migrations between 
regions that have a satisfying control of the epidemic.
\begin{figure}[ht]
\centering
\includegraphics[scale=0.55]{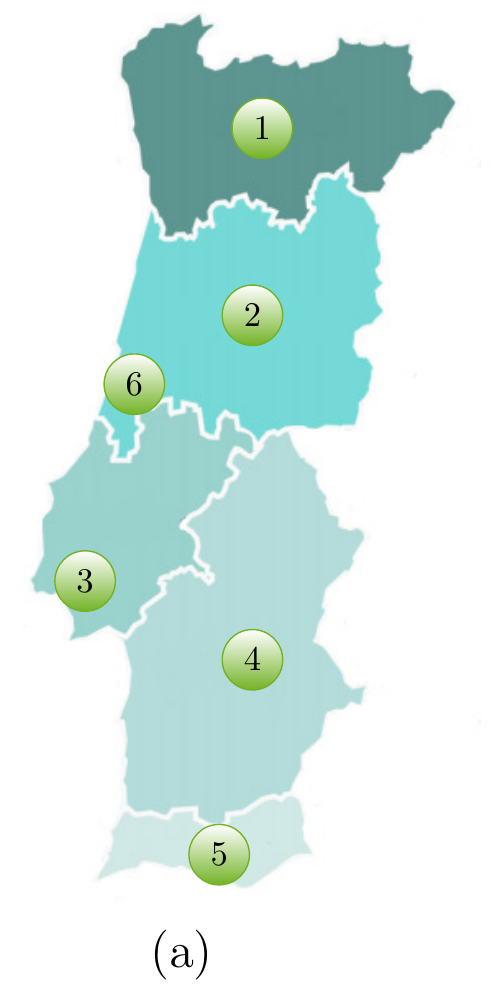}
\includegraphics[scale=0.55]{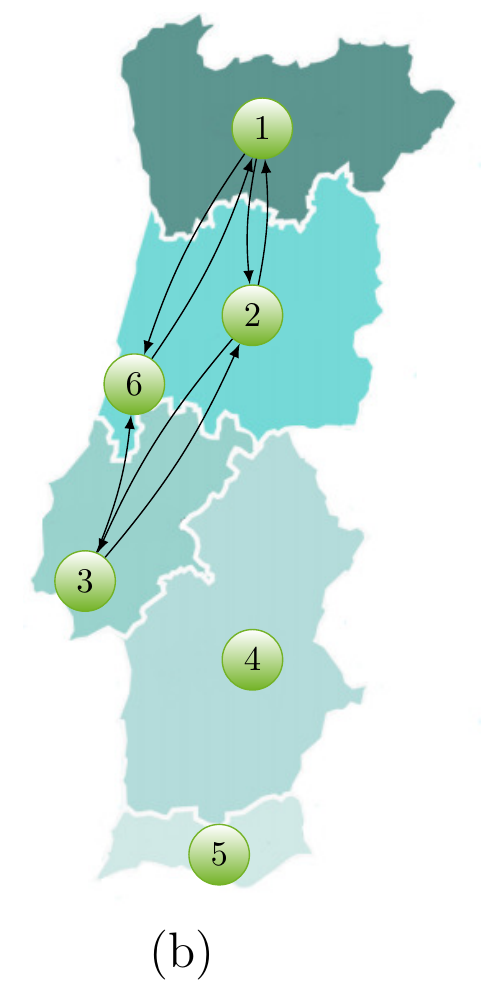}
\includegraphics[scale=0.55]{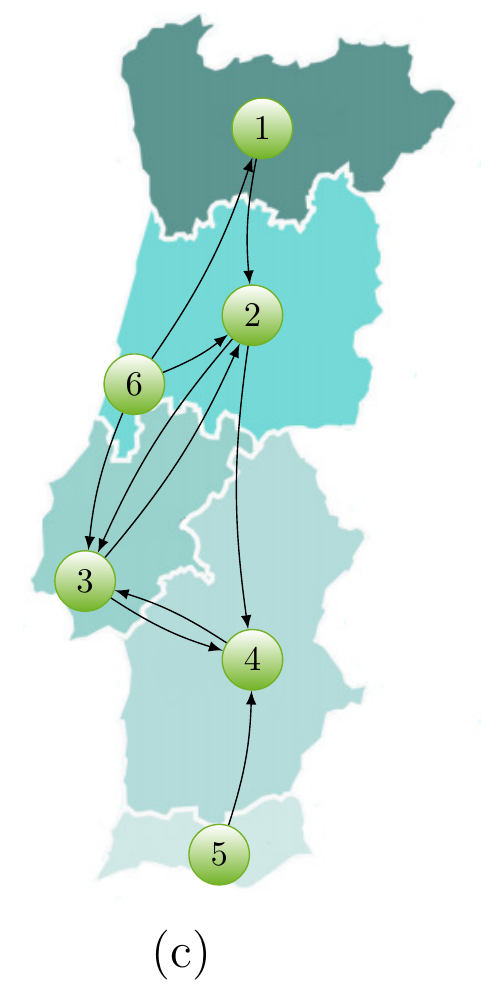}
\includegraphics[scale=0.55]{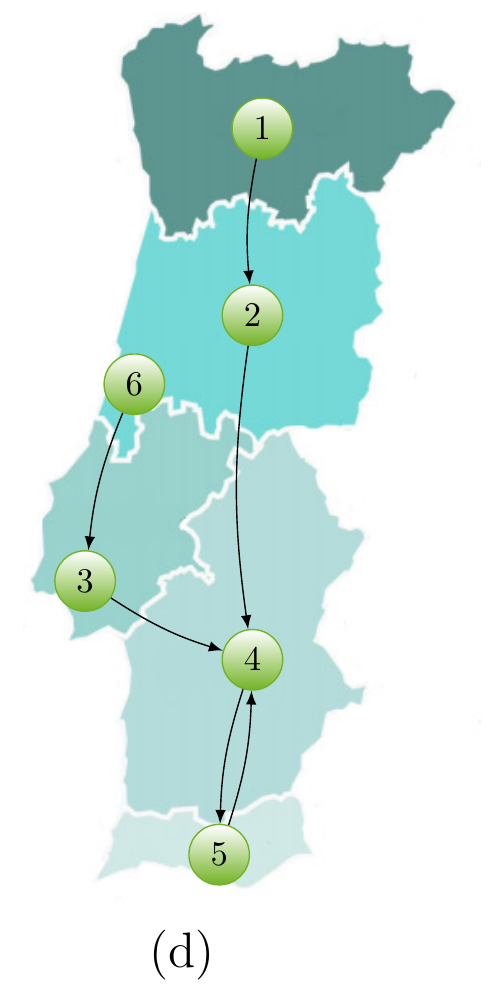}
\caption{Four remarkable topologies.
(a) Empty topology, which corresponds to the situation 
where individuals do not migrate from one region to another.
(b) Topology that minimizes the level of infection.
(c) Topology that leads to a level of infection greater than the level 
of the empty topology for only a weak coupling strength.
(d) Topology that permanently overcomes the level of the empty topology.}
\label{fig:four-topologies}
\end{figure}


\section{Discussion}
\label{sec:discussion}

One of the factors associated with the spread of the pandemic is the mobility of populations. 
It was based on this evidence that, in the first phase of the pandemic response, a generalized 
containment of the population was implemented in a big number of countries. As a result of this 
restriction in the mobility, it was possible to obtain a decrease in the number of infected.
In the numerical simulations considered in this work, for the complex network problem 
\eqref{eq:SAIRP-network}, it is assumed that only susceptible ($S$) and asymptomatic infected ($A$) 
are circulating among the 6 identified regions of mainland Portugal. Starting from the number 
of active infected people ($I$) for the empty topology, in which it is assumed that there are 
no flows between regions, the results of the simulations indicate that there are more topologies 
that increase the number of active infected people than those that cause this number to decrease. 
Hence, this result reinforces the need to reduce circulation between regions as one of the tools 
to reduce the spread of the virus and decrease the number of active infected people. Still, 
it is interesting to note that there are other topologies (around 1 out of 3) in which the number 
of active infected individuals decreases with the movement of people between regions. 
This finding can be seen in Figure~\ref{fig:sample-random-topologies}, where the average number 
of active infected individuals for the 1000 tested topologies is shown. According to the topology 
that minimizes the number of active infected individuals, whose representation is shown 
in Figure~\ref{fig:four-topologies}~b), it is possible to verify that the strategy of preventing 
the flow to the regions of \emph{Alentejo} (4) and \emph{Algarve} (5) contributes to reduce 
the number of active infected. This scenario was to be expected as those regions have a much 
lower number of cases than the rest of Portugal. Bidirectional circulation between 
\emph{Norte} (1) and \emph{Centro} (2) regions leads to a reduction in the number 
of active infected people (if one compares with the empty topology), but the movement between 
\emph{Pinhal Litoral} region (6) and \emph{Lisboa e Vale do Tejo} region (3) should 
only be done in the direction to prevent circulation from the region where there are less 
infected (6) to the region (3) that has a high number of active cases. 
In Figure~\ref{fig:four-topologies}~c) it is possible to verify that the existence 
of this connection in the direction (6) to (3) increases the number of active cases.
The topologies involving the regions of \emph{Alentejo} (4) and \emph{Algarve} (5) 
increase the number of cases. One extreme example of such topology can be found in 
Figure~\ref{fig:four-topologies}~d). In this case, the direction of the flows presents 
a topology that converges to the \emph{Alentejo} (4) and it is expected that this 
will translate into an increase in the number of infected (even in the case where 
the connection strength is small, as shown in Figure~\ref{fig:four-topologies}~c) 
and Figure~\ref{fig:number-infected-topology-wrt-topology} at right), as it is the 
region where the number of cases is lesser. The way the epidemic evolves is different 
between the different regions of a country, but also between countries. These specific 
characteristics of context (cultural, demographic, economic, \ldots) seen in the number 
of active infected individuals (which is not uniform across regions), can be modeled 
by pseudo-periodic piecewise functions, as proved in this work. Accordingly, the 
application of such methodology to identify flows between regions or countries 
is a tool with enormous potential in the current pandemic context, and can be applied 
in the management of outbreaks (in regional terms) but also to manage 
the opening/closing of borders.


\section{Conclusion and future work}
\label{sec:conclusion}

The way the epidemic evolves is different between the different regions of a country, 
but also between countries. These specific characteristics of context (cultural, 
demographic, economic, \ldots) seen in the number of active infected individuals, 
which is not uniform across regions, can be modeled by pseudo-periodic piecewise functions, 
as proved in this work. Accordingly, the application of such methodology to identify flows 
between regions or countries is a tool with enormous potential in the current pandemic context, 
and can be applied in the management of outbreaks (in regional terms) but also to manage 
the opening/closing of borders. In the implemented simulations, the intensity of the flows 
($\sigma_S$ and $\sigma_A$) are assumed to be equal. Still, and as future work, it would 
be useful, in terms of epidemiological management, to be able to model the intensity 
of flows by class (Asymptomatic, Susceptible, Active Infected, Recovered, and Protected) 
between regions.


\section*{Acknowledgments} 

This research is partially supported by the Portuguese Foundation for Science and Technology (FCT) 
within ``Project Nr.~147 -- Controlo \'Otimo e Modela\c{c}\~ao Matem\'atica da Pandemia \text{COVID-19}: 
contributos para uma estrat\'egia sist\'emica de interven\c{c}\~ao em sa\'ude na comunidade'', 
in the scope of the ``RESEARCH 4 COVID-19'' call financed by FCT, and by project UIDB/04106/2020 (CIDMA).
Silva is also supported by national funds (OE), through FCT, I.P., in the scope of the framework contract 
foreseen in the numbers 4, 5 and 6 of the article 23, of the Decree-Law 57/2016, of August 29, 
changed by Law 57/2017, of July 19. Rui Fonseca-Pinto and Rui Passadouro are supported by the FCT 
projects UIDB/05704/2020 and UIDP/05704/2020.



\end{document}